\documentclass[global]{journal}

\smartqed

\usepackage{graphicx}
\usepackage{mathptmx}
\usepackage{ispo}

\journalname{Acta Informatica}

\begin{document}

\sloppy

\title{Instruction Sequence Processing Operators}

\author{J.A. Bergstra \and C.A. Middelburg}

\institute{Informatics Institute, Faculty of Science,
           University of Amsterdam, \\
           Science Park~904, 1098~XH Amsterdam, the Netherlands \\
           \email{J.A.Bergstra@uva.nl, C.A.Middelburg@uva.nl}}

\date{}

\maketitle

\begin{abstract}
Instruction sequence is a key concept in practice, but it has as yet
not come prominently into the picture in theoretical circles.
This paper concerns instruction sequences, the behaviours produced by
them under execution, the interaction between these behaviours and
components of the execution environment, and two issues relating to
computability theory.
Positioning Turing's result regarding the undecidability of the halting
problem as a result about programs rather than machines, and taking
instruction sequences as programs, we analyse the autosolvability
requirement that a program of a certain kind must solve the halting
problem for all programs of that kind.
We present novel results concerning this autosolvability requirement.
The analysis is streamlined by using the notion of a functional unit,
which is an abstract state-based model of a machine.
In the case where the behaviours exhibited by a component of an
execution environment can be viewed as the behaviours of a machine in
its different states, the behaviours concerned are completely
determined by a functional unit.
The above-mentioned analysis involves functional units whose possible
states represent the possible contents of the tapes of Turing machines
with a particular tape alphabet.
We also investigate functional units whose possible states are the
natural numbers.
This investigation yields a novel computability result, viz.\ the
existence of a universal computable functional unit for natural
numbers.
\end{abstract}

\section{Introduction}
\label{sect-intro}

The concept of an instruction sequence is a very primitive concept in
computing.
Instruction sequence execution has always been part of computing
because of the fact that it underlies virtually all past and current
generations of computers.
It happens that, given a precise definition of an appropriate notion
of an instruction sequence, many issues in computer science can be
clearly explained in terms of instruction sequences, from issues of a
computer-architectural kind to issues of a computation-theoretic kind.
A simple yet interesting example is that a program can simply be
defined as a text that denotes an instruction sequence.
Such a definition corresponds to an appealing empirical perspective
found among practitioners.

In theoretical computer science, the meaning of programs usually plays a
prominent part in the explanation of many issues concerning programs.
Moreover, what is taken for the meaning of programs is mathematical by
nature.
On the other hand, it is customary that practitioners do not fall back
on the mathematical meaning of programs in case explanation of issues
concerning programs is needed.
They phrase their explanations from an empirical perspective.
An appealing empirical perspective is the one that a program is in
essence an instruction sequence and an instruction sequence under
execution produces a behaviour that is controlled by its execution
environment: each step performed actuates the processing of an
instruction by the execution environment and a reply returned at
completion of the processing determines how the behaviour proceeds.

The work presented in this paper belongs to a line of research which
started with an attempt to approach the semantics of programming
languages from the perspective mentioned above.
The first published paper on this approach is~\cite{BL00a}.
That paper is superseded by~\cite{BL02a} with regard to the groundwork
for the approach: program algebra, an algebraic theory of single-pass
instruction sequences, and basic thread algebra, an algebraic theory of
mathematical objects that represent in a direct way the behaviours
produced by instruction sequences under execution.%
\footnote
{In~\cite{BL02a}, basic thread algebra is introduced under the name
 basic polarized process algebra.
}
The main advantages of the approach are that it does not require a lot
of mathematical background and that it is more appealing to
practitioners than the main approaches to programming language
semantics: the operational approach, the denotational approach and the
axiomatic approach.
For an overview of these approaches, see e.g.~\cite{Mos06a}.

As a continuation of the work on a new approach to programming language
semantics, the notion of an instruction sequence was subjected to
systematic and precise analysis using the groundwork laid earlier.
This led among other things to expressiveness results about the
instruction sequences considered and variations of the instruction
sequences considered (see e.g.~\cite{PZ06a,BM07g,BP09a,BM08h,BM08b}).
Instruction sequences are under discussion for many years in diverse
work on computer architecture, as witnessed by
e.g.~\cite{Lun77a,PD80a,HJP82a,Bak91a,XT96a,BH97a,NH97a,OH00a,TW07a},
but the notion of an instruction sequence has never been subjected to
any precise analysis.

As another continuation of the work on a new approach to programming
language semantics, selected issues relating to well-known subjects
from the theory of computation and the area of computer architecture
were rigorously investigated thinking in terms of instruction sequences
(see e.g.~\cite{BM08g,BM09i,BM07c,BM11c}).
The subjects from the theory of computation, namely the halting problem
and non-uniform computational complexity, are usually investigated
thinking in terms of a common model of computation such as Turing
machines and Boolean circuits (see e.g.~\cite{Her65a,Sip06a,AB09a}).
The subjects from the area of computer architecture, namely instruction
sequence performance, instruction set architectures and remote
instruction processing, are usually not investigated in a rigorous way
at all.

This paper concerns among other things an investigation of issues
relating to the halting problem thinking in terms of instruction
sequences.
Positioning Turing's result regarding the undecidability of the halting
problem (see e.g.~\cite{Tur37a}) as a result about programs rather than
machines, and taking instruction sequences as programs, we analyse the
autosolvability requirement that a program of a certain kind must solve
the halting problem for all programs of that kind.
We present positive and negative results concerning the autosolvability
of the halting problem for programs.
To our knowledge, these results are new and unusual.

Most work done in the line of research sketched above requires that
basic thread algebra, i.e.\ the algebraic theory of mathematical
objects that represent in a direct way the behaviours produced by
instruction sequences under execution, is extended to deal with the
interaction between instruction sequences under execution and
components of their execution environment concerning the processing of
instructions.
The first published paper on such an extended theory is~\cite{BP02a}.
A substantial re-design of the extended theory presented in that paper
is presented in the current paper.
The changes introduced allow for the material from quite a part of the
work done in the line of research sketched above to be streamlined.

Further streamlining is achieved in this paper by introducing and using
the notion of a functional unit.
In the extended theory, a rather abstract view of the behaviours
exhibited by components of execution environments is taken.
The view is just sufficiently concrete for the purpose of the theory.
A functional unit is an abstract model of a machine.
Under the abstract view of the behaviours exhibited by a component of
an execution environment, the behaviours concerned are completely
determined by a functional unit in the frequently occurring case that
they can be viewed as the behaviours of a machine in its different
states.
The current paper also concerns an investigation of functional units
whose possible states are the natural numbers.
This investigation yields a computability result that is new and
unusual as far as we know, namely the existence of a universal
computable functional unit for natural numbers.

The investigations carried out in the line of research sketched above
demonstrate that the concept of an instruction sequence offers a novel
and useful viewpoint on issues relating to diverse subjects.
In view of the very primitive nature of this concept, it is in fact
rather surprising that instruction sequences have never been a theme in
computer science.
A theoretical understanding of issues in terms of instruction sequences
will probably become increasingly more important to a growing number of
developments in computer science.
Among them are for instance the developments with respect to techniques
for high-performance program execution on classical or non-classical
computers and techniques for estimating execution times of hard
real-time systems.
For these and other such developments, the abstractions usually made do
not allow for all relevant details to be considered.

Some marginal notes are in order.
In this paper, we use an extension of a program notation rooted in
program algebra instead of an extension of program algebra itself.
The program notation in question has been chosen because it turned out
to be appropriate.
However, in principle any program notation that is as expressive as the
closed terms of program algebra would do.
The above-mentioned analysis of the autosolvability requirement
inherent in Turing's result regarding the undecidability of the halting
problem involves functional units whose possible states are objects
that represent the possible contents of the tapes of Turing machines
with a particular tape alphabet.

Henceforth, objects that represent in a direct way the behaviours
produced by instruction sequences under execution are called threads,
objects that represent the behaviours exhibited by components of
execution environments are called services, and collections of named
services are called service families.
In order to deal with the different aspects of the interaction between
instruction sequences under execution and components of their execution
environment concerning the processing of instructions, three operators
are added to basic thread algebra.
Because these operators are primarily intended to be used to describe
and analyse instruction sequence processing, they are loosely referred
to by the term instruction sequence processing operators.

This paper is organized as follows.
First, we give a survey of the program notation used in this paper
(Section~\ref{sect-PGLBbt}) and define its semantics using basic thread
algebra (Section~\ref{sect-BTAbt}).
Next, we introduce services and a composition operator for services
families (Section~\ref{sect-SF}), and the three operators that are
related to the processing of instructions by a service family
(Section~\ref{sect-TSI}).
Then, we propose to comply with conventions that exclude the use of
terms that are not really intended to denote anything
(Sections~\ref{sect-RUC}).
After that, we give an example related to the processing of instructions
by a service family (Section~\ref{sect-example}).
Further, we present an interesting variant of one of the above-mentioned
operators related to the processing of instructions
(Section~\ref{sect-abstr-use}).
Thereafter, we introduce the concept of a functional unit and related
concepts (Section~\ref{sect-func-unit}).
Subsequently, we investigate functional units for natural numbers
(Section~\ref{sect-func-unit-nat}).
Then, we define autosolvability and related notions in terms of
functional units related to Turing machine tapes
(Section~\ref{sect-func-unit-sbs}).
After that, we discuss the weakness of interpreters when it comes to
solving the halting problem (Section~\ref{sect-interpreters}) and give
positive and negative results concerning the autosolvability of the
halting problem (Section~\ref{sect-autosolvability}).
Finally, we make some concluding remarks (Section~\ref{sect-concl}).

This paper consolidates material from the
reports~\cite{BM09m,BM09l,BM09k}.

\section{PGLB with Boolean Termination}
\label{sect-PGLBbt}

In this section, we introduce the program notation \PGLBbt\ (PGLB with
Boolean termination).
In~\cite{BL02a}, a hierarchy of program notations rooted in program
algebra is presented.
One of the program notations that belong to this hierarchy is \PGLB\
(ProGramming Language B).
This program notation is close to existing assembly languages and has
relative jump instructions.
\PGLBbt\ is \PGLB\ extended with two termination instructions that allow
for the execution of instruction sequences to yield a Boolean value at
termination.
The extension makes it possible to deal naturally with instruction
sequences that implement some test, which is relevant throughout the
paper.

In \PGLBbt, it is assumed that a fixed but arbitrary non-empty finite
set $\BInstr$ of \emph{basic instructions} has been given.
The intuition is that the execution of a basic instruction in most
instances modifies a state and in all instances produces a reply at its
completion.
The possible replies are $\True$ (standing for true) and $\False$
(standing for false), and the actual reply is in most instances
state-dependent.
Therefore, successive executions of the same basic instruction may
produce different replies.
The set $\BInstr$ is the basis for the set of all instructions that may
appear in the instruction sequences considered in \PGLBbt.
These instructions are called primitive instructions.

\PGLBbt\ has the following primitive instructions:
\begin{itemize}
\item
for each $a \in \BInstr$, a \emph{plain basic instruction} $a$;
\item
for each $a \in \BInstr$, a \emph{positive test instruction} $\ptst{a}$;
\item
for each $a \in \BInstr$, a \emph{negative test instruction} $\ntst{a}$;
\item
for each $l \in \Nat$, a \emph{forward jump instruction}
$\fjmp{l}$;
\item
for each $l \in \Nat$, a \emph{backward jump instruction}
$\bjmp{l}$;
\item
a \emph{plain termination instruction} $\halt$;
\item
a \emph{positive termination instruction} $\haltP$;
\item
a \emph{negative termination instruction} $\haltN$.
\end{itemize}
\PGLBbt\ instruction sequences have the form
$u_1 \conc \ldots \conc u_k$, where $u_1,\ldots,u_k$ are primitive
instructions of \PGLBbt.

On execution of a \PGLBbt\ instruction sequence, these primitive
instructions have the following effects:
\begin{itemize}
\item
the effect of a positive test instruction $\ptst{a}$ is that basic
instruction $a$ is executed and execution proceeds with the next
primitive instruction if $\True$ is produced and otherwise the next
primitive instruction is skipped and execution proceeds with the
primitive instruction following the skipped one -- if there is no
primitive instructions to proceed with, deadlock occurs;
\item
the effect of a negative test instruction $\ntst{a}$ is the same as the
effect of $\ptst{a}$, but with the role of the value produced reversed;
\item
the effect of a plain basic instruction $a$ is the same as the effect of
$\ptst{a}$, but execution always proceeds as if $\True$ is produced;
\item
the effect of a forward jump instruction $\fjmp{l}$ is that execution
proceeds with the $l^\mathrm{th}$ next primitive instruction -- if $l$
equals $0$ or there is no primitive instructions to proceed with,
deadlock occurs;
\item
the effect of a backward jump instruction $\bjmp{l}$ is that execution
proceeds with the $l^\mathrm{th}$ previous primitive instruction -- if
$l$ equals $0$ or there is no primitive instructions to proceed with,
deadlock occurs;
\item
the effect of the plain termination instruction $\halt$ is that
execution terminates and in doing so does not deliver a value;
\item
the effect of the positive termination instruction $\haltP$ is that
execution terminates and in doing so delivers the Boolean value $\True$;
\item
the effect of the negative termination instruction $\haltN$ is that
execution terminates and in doing so delivers the Boolean value
$\False$.
\end{itemize}

A simple example of a \PGLBbt\ instruction sequence is
\begin{ldispl}
\ptst{a} \conc \fjmp{2} \conc \bjmp{2} \conc b \conc \haltP\;.
\end{ldispl}%
On execution of this instruction sequence, first the basic instruction
$a$ is executed repeatedly until its execution produces the reply
$\True$, next the basic instruction $b$ is executed, and after that
execution terminates with delivery of the value $\True$.

From Section~\ref{sect-func-unit}, we will use a restricted version of
\PGLBbt\ called \PGLBsbt\ (\PGLB\ with strict Boolean termination).
The primitive instructions of \PGLBsbt\ are the primitive instructions
of \PGLBbt\ with the exception of the plain termination instruction.
Thus, \PGLBsbt\ instruction sequences are \PGLBbt\ instruction sequences
in which the plain termination instruction does not occur.

In Section~\ref{sect-example}, we will give examples of instruction
sequences for which the delivery of a Boolean value at termination of
their execution is natural.
There, we will write $\Conc{i = 1}{n} P_i$, where $P_1,\ldots,P_n$ are
\PGLBbt\ instruction sequences, for the \PGLBbt\ instruction sequence
$P_1 \conc \ldots \conc P_n$.

\section{Thread Extraction}
\label{sect-BTAbt}

In this section, we make precise in the setting of \BTAbt\ (Basic Thread
Algebra with Boolean termination) which behaviours are exhibited on
execution by \PGLBbt\ instruction sequences.
We start by introducing \BTAbt.
In~\cite{BL02a}, \BPPA\ (Basic Polarized Process Algebra) is introduced
as a setting for modelling the behaviours exhibited by instruction
sequences under execution.
Later, \BPPA\ has been renamed to \BTA\ (Basic Thread Algebra).
\BTAbt\ is \BTA\ extended with two constants for termination at which a
Boolean value is yielded.

In \BTAbt, it is assumed that a fixed but arbitrary non-empty finite set
$\BAct$ of \emph{basic actions}, with $\Tau \not\in \BAct$, has been
given.
We write $\BActTau$ for $\BAct \union \set{\Tau}$.
The members of $\BActTau$ are referred to as \emph{actions}.

A thread is a behaviour which consists of performing actions in a
sequential fashion.
Upon each basic action performed, a reply from an execution environment
determines how the thread proceeds.
The possible replies are the Boolean values $\True$ and $\False$.
Performing the action $\Tau$ will always lead to the reply $\True$.

\BTAbt\ has one sort: the sort $\Thr$ of \emph{threads}.
We make this sort explicit because we will extend \BTAbt\ with
additional sorts in Section~\ref{sect-TSI}.
To build terms of sort $\Thr$, \BTAbt\ has the following constants and
operators:
\begin{itemize}
\item
the \emph{deadlock} constant $\const{\DeadEnd}{\Thr}$;
\item
the \emph{plain termination} constant $\const{\Stop}{\Thr}$;
\item
the \emph{positive termination} constant $\const{\StopP}{\Thr}$;
\item
the \emph{negative termination} constant $\const{\StopN}{\Thr}$;
\item
for each $a \in \BActTau$, the binary \emph{postconditional composition}
operator $\funct{\pcc{\ph}{a}{\ph}}{\Thr \x \Thr}{\Thr}$.
\end{itemize}
We assume that there is a countably infinite set of variables of sort
$\Thr$ which includes $x,y,z$.
Terms of sort $\Thr$ are built as usual.
We use infix notation for postconditional composition.
We introduce \emph{action prefixing} as an abbreviation: $a \bapf p$,
where $p$ is a term of sort $\Thr$, abbreviates $\pcc{p}{a}{p}$.

The thread denoted by a closed term of the form $\pcc{p}{a}{q}$ will
first perform $a$, and then proceed as the thread denoted by $p$
if the reply from the execution environment is $\True$ and proceed as
the thread denoted by $q$ if the reply from the execution environment is
$\False$.
The thread denoted by $\DeadEnd$ will become inactive, the thread
denoted by $\Stop$ will terminate without yielding a value, and the
threads denoted by $\StopP$ and $\StopN$ will terminate and with that
yield the Boolean values $\True$ and $\False$, respectively.

A simple example of a closed \BTAbt\ term is
\begin{ldispl}
\pcc{(b \bapf \StopP)}{a}{(c \bapf \StopN)}\;.
\end{ldispl}%
This term denotes the thread that first performs basic action $a$, if
the reply from the execution environment on performing $a$ is $\True$,
next performs the basic action $b$ and after that terminates with
delivery of the Boolean value $\True$, and if the reply from the
execution environment on performing $a$ is $\False$, next performs the
basic action $c$ and after that terminates with delivery of the Boolean
value $\False$.

\BTAbt\ has only one axiom.
This axiom is given in Table~\ref{axioms-BTAbt}.%
\begin{table}[!t]
\caption{Axiom of \BTAbt}
\label{axioms-BTAbt}
\begin{eqntbl}
\begin{axcol}
\pcc{x}{\Tau}{y} = \pcc{x}{\Tau}{x}                     & \axiom{T1}
\end{axcol}
\end{eqntbl}
\end{table}

Each closed \BTA\ term denotes a finite thread, i.e.\ a thread with a
finite upper bound to the number of actions that it can perform.
Infinite threads, i.e.\ threads without a finite upper bound to the
number of actions that it can perform, can be described by guarded
recursion.

A \emph{guarded recursive specification} over \BTAbt\ is a set of
recursion equations $E = \set{x = t_x \where x \in V}$, where $V$ is a
set of variables of sort $\Thr$ and each $t_x$ is a \BTAbt\ term of the
form $\DeadEnd$, $\Stop$, $\StopP$, $\StopN$ or $\pcc{t}{a}{t'}$ with
$t$ and $t'$ that contain only variables from $V$.
We are only interested in models of \BTAbt\ in which guarded recursive
specifications have unique solutions, such as the appropriate expansion
of the projective limit model of \BTA\ presented in~\cite{BB03a}.

A simple example of a guarded recursive specification is the one
consisting of the following two equations:
\begin{ldispl}
x = \pcc{x}{a}{y}\;, \qquad
y = \pcc{y}{b}{\Stop}\;.
\end{ldispl}%
The $x$-component of the solution of this guarded recursive
specification is the thread that first performs basic action $a$
repeatedly until the reply from the execution environment on performing
$a$ is $\False$, next performs basic action $b$ repeatedly until the
reply from the execution environment on performing $b$ is $\False$, and
after that terminates without delivery of a Boolean value.

To reason about infinite threads, we assume the infinitary conditional
equation AIP (Approximation Induction Principle).
AIP is based on the view that two threads are identical if their
approximations up to any finite depth are identical.
The approximation up to depth $n$ of a thread is obtained by cutting it
off after it has performed $n$ actions.
In AIP, the approximation up to depth $n$ is phrased in terms of the
unary \emph{projection} operator $\funct{\projop{n}}{\Thr}{\Thr}$.
AIP and the axioms for the projection operators are given in
Table~\ref{axioms-AIP}.%
\begin{table}[!t]
\caption{Approximation induction principle}
\label{axioms-AIP}
\begin{eqntbl}
\begin{axcol}
\AND{n \geq 0}{} \proj{n}{x} = \proj{n}{y} \Implies
                                                  x = y & \axiom{AIP} \\
\proj{0}{x} = \DeadEnd                                  & \axiom{P0} \\
\proj{n+1}{\StopP} = \StopP                             & \axiom{P1a} \\
\proj{n+1}{\StopN} = \StopN                             & \axiom{P1b} \\
\proj{n+1}{\Stop} = \Stop                               & \axiom{P1c} \\
\proj{n+1}{\DeadEnd} = \DeadEnd                         & \axiom{P2} \\
\proj{n+1}{\pcc{x}{a}{y}} =
                      \pcc{\proj{n}{x}}{a}{\proj{n}{y}} & \axiom{P3}
\end{axcol}
\end{eqntbl}
\end{table}
In this table, $a$ stands for an arbitrary action from $\BActTau$ and
$n$ stands for an arbitrary natural number.

We can prove that the projections of solutions of guarded recursive
specifications over \BTAbt\ are representable by closed \BTAbt\ terms of
sort $\Thr$.
\begin{lemma}
\label{lemma-projections-BTAbt}
Let $E$ be a guarded recursive specification over \BTAbt, and
let $x$ be a variable occurring in $E$.
Then, for all $n \in \Nat$, there exists a closed \BTAbt\ term $p$ of
sort $\Thr$ such that $\proj{n}{x} = p$ is derivable from $E$ and the
axioms for the projection operators.
\end{lemma}
\begin{proof}
In the case of \BTA, this is proved in~\cite{BM06a} as part of the proof
of Theorem~1 from that paper.
The proof concerned goes through in the case of \BTAbt.
\qed
\end{proof}
For example, let $E$ be the guarded recursive specification
consisting of the equation $x = \pcc{x}{a}{\Stop}$ only.
Then the projections of $x$ are as follows:
\begin{ldispl}
\begin{aeqns}
\proj{0}{x} & = & \DeadEnd\;,
\\
\proj{1}{x} & = & \pcc{\DeadEnd}{a}{\Stop}\;,
\\
\proj{2}{x} & = & \pcc{(\pcc{\DeadEnd}{a}{\Stop})}{a}{\Stop}\;,
\\
\proj{3}{x} & = &
\pcc{(\pcc{(\pcc{\DeadEnd}{a}{\Stop})}{a}{\Stop})}{a}{\Stop}\;,
\\ & \vdots &
\end{aeqns}
\end{ldispl}%

Henceforth, we will write \BTAbtp\ for \BTAbt\ extended with the
projection operators, the axioms for the projection operators, and AIP.

The behaviours exhibited on execution by \PGLBbt\ instruction sequences
are considered to be threads, with the basic instructions taken for
basic actions.
The \emph{thread extraction} operation $\extr{\ph}$ defines, for each
\PGLBbt\ instruction sequence, the behaviour exhibited on its execution.
The thread extraction operation is defined by
$\extr{u_1 \conc \ldots \conc u_k} =
 \extr{1,u_1 \conc \ldots \conc u_k}$,
where $\extr{\ph,\ph}$ is defined by the equations given in
Table~\ref{axioms-thread-extr} (for $a \in \BInstr$ and $l,i \in \Nat$)%
\begin{table}[!t]
\caption{Defining equations for the thread extraction operation}
\label{axioms-thread-extr}
\begin{eqntbl}
\begin{aceqns}
\extr{i,u_1 \conc \ldots \conc u_k} & = & \DeadEnd
& \mif \mathrm{not}\; 1 \leq i \leq k \\
\extr{i,u_1 \conc \ldots \conc u_k} & = &
a \bapf \extr{i+1,u_1 \conc \ldots \conc u_k}
& \mif u_i = a \\
\extr{i,u_1 \conc \ldots \conc u_k} & = &
\pcc{\extr{i+1,u_1 \conc \ldots \conc u_k}}{a}
    {\extr{i+2,u_1 \conc \ldots \conc u_k}}
& \mif u_i = +a \\
\extr{i,u_1 \conc \ldots \conc u_k} & = &
\pcc{\extr{i+2,u_1 \conc \ldots \conc u_k}}{a}
    {\extr{i+1,u_1 \conc \ldots \conc u_k}}
& \mif u_i = -a \\
\extr{i,u_1 \conc \ldots \conc u_k} & = &
\extr{i+l,u_1 \conc \ldots \conc u_k}
& \mif u_i = \fjmp{l} \\
\extr{i,u_1 \conc \ldots \conc u_k} & = &
\extr{i \monus l,u_1 \conc \ldots \conc u_k}
& \mif u_i = \bjmp{l} \\
\extr{i,u_1 \conc \ldots \conc u_k} & = & \Stop
& \mif u_i = \halt \\
\extr{i,u_1 \conc \ldots \conc u_k} & = & \StopP
& \mif u_i = \haltP \\
\extr{i,u_1 \conc \ldots \conc u_k} & = & \StopN
& \mif u_i = \haltN
\end{aceqns}
\end{eqntbl}
\end{table}%
\footnote
{As usual, we write $i \monus j$ for the monus of $i$ and $j$, i.e.\
 $i \monus j = i - j$ if $i \geq j$ and $i \monus j = 0$ otherwise.
}
and the rule that $\extr{i,u_1 \conc \ldots \conc u_k} = \DeadEnd$ if
$u_i$ is the beginning of an infinite jump chain.%
\footnote
{This rule can be formalized, cf.~\cite{BM07g}.}

If $1 \leq i \leq k$, $\extr{i,u_1 \conc \ldots \conc u_k}$ can be read
as the behaviour exhibited by $u_1 \conc \ldots \conc u_k$ on execution
if execution starts at the $i^\mathrm{th}$ primitive instruction, i.e.\
$u_i$.
By default, execution starts at the first primitive instruction.

Some simple examples of thread extraction are
\begin{ldispl}
\begin{aeqns}
\extr{\ptst{a} \conc \fjmp{2} \conc \fjmp{3} \conc b \conc \haltP}
& = &
\pcc{(b \bapf \StopP)}{a}{\DeadEnd}\;,
\\
\extr{\ptst{a} \conc \ntst{b} \conc c \conc \halt}
& = &
\pcc{(\pcc{\Stop}{b}{(c \bapf \Stop)})}{a}{(c \bapf \Stop)}\;.
\end{aeqns}
\end{ldispl}%
The behaviour exhibited on execution may also be an infinite thread.
For example,
\begin{ldispl}
\extr{a \conc
      \ptst{b} \conc \fjmp{2} \conc \fjmp{3} \conc
      c \conc \fjmp{4} \conc
      \ntst{d} \conc \halt \conc a \conc \bjmp{8}}
\end{ldispl}%
denotes the $x$-component of the solution of the guarded recursive
specification consisting of the following two equations:
\begin{ldispl}
x = a \bapf y\;, \qquad
y = \pcc{(c \bapf y)}{b}{(\pcc{x}{d}{\Stop})}\;.
\end{ldispl}%

\section{Services and Service Families}
\label{sect-SF}

In this section, we introduce service families and a composition
operator for service families.
We start by introducing services.

It is assumed that a fixed but arbitrary non-empty finite set $\Meth$
of \emph{methods} has been given.
A service is able to process certain methods.
The processing of a method may involve a change of the service.
At completion of the processing of a method, the service produces a
reply value.
The set $\Replies$ of \emph{reply values} is the $3$-element set
$\set{\True,\False,\Div}$.
The reply value $\Div$ stands for divergent.

For example, a service may be able to process methods for pushing a
natural number on a stack ($\push{n}$), testing whether the top of the
stack equals a natural number ($\topeq{n}$), and popping the top element
from the stack ($\pop$).
Processing of a pushing method or a popping method changes the service,
because it changes the stack with which it deals, and produces the reply
value $\True$ if no stack overflow or stack underflow occurs and
$\False$ otherwise.
Processing of a testing method does not change the service, because it
does not changes the stack with which it deals, and produces the reply
value $\True$ if the test succeeds and $\False$ otherwise.
Attempted processing of a method that the service is not able to process
changes the service into one that is not able to process any method and
produces the reply $\Div$.

In \SFA, the algebraic theory of service families introduced below, the
following is assumed with respect to services:
\begin{itemize}
\item
a set $\Services$ of services has been given together with:
\begin{itemize}
\item
for each $m \in \Meth$,
a total function $\funct{\effect{m}}{\Services}{\Services}$;
\item
for each $m \in \Meth$,
a total function $\funct{\sreply{m}}{\Services}{\Replies}$;
\end{itemize}
satisfying the condition that there exists a unique $S \in \Services$
with $\effect{m}(S) = S$ and $\sreply{m}(S) = \Div$ for all
$m \in \Meth$;
\item
a signature $\Sig{\Services}$ has been given that includes the following
sort:
\begin{itemize}
\item
the sort $\Serv$ of \emph{services};
\end{itemize}
and the following constant and operators:
\begin{itemize}
\item
the \emph{empty service} constant $\const{\emptyserv}{\Serv}$;
\item
for each $m \in \Meth$,
the \emph{derived service} operator
$\funct{\derive{m}}{\Serv}{\Serv}$;
\end{itemize}
\item
$\Services$ and $\Sig{\Services}$ are such that:
\begin{itemize}
\item
each service in $\Services$ can be denoted by a closed term of sort
$\Serv$;
\item
the constant $\emptyserv$ denotes the unique $S \in \Services$ such
that $\effect{m}(S) = S$ and $\sreply{m}(S) = \Div$ for all
$m \in \Meth$;
\item
if closed term $t$ denotes service $S$, then $\derive{m}(t)$ denotes
service $\effect{m}(S)$.
\end{itemize}
\end{itemize}

When a request is made to service $S$ to process method $m$:
\begin{itemize}
\item
if $\sreply{m}(S) \neq \Div$, then $S$ processes $m$, produces the reply
$\sreply{m}(S)$, and next proceeds as $\effect{m}(S)$;
\item
if $\sreply{m}(S) = \Div$, then $S$ is not able to process method $m$
and proceeds as $\emptyserv$.
\end{itemize}
The empty service $\emptyserv$ is unable to process any method.

It is also assumed that a fixed but arbitrary non-empty finite set
$\Foci$ of \emph{foci} has been given.
Foci play the role of names of services in the service family offered by
an execution environment.
A service family is a set of named services where each name occurs only
once.

\SFA\ has the sorts, constants and operators in $\Sig{\Services}$
and in addition the following sort:
\begin{itemize}
\item
the sort $\ServFam$ of \emph{service families};
\end{itemize}
and the following constant and operators:
\begin{itemize}
\item
the \emph{empty service family} constant $\const{\emptysf}{\ServFam}$;
\item
for each $f \in \Foci$, the unary \emph{singleton service family}
operator $\funct{\mathop{f{.}} \ph}{\Serv}{\ServFam}$;
\item
the binary \emph{service family composition} operator
$\funct{\ph \sfcomp \ph}{\ServFam \x \ServFam}{\ServFam}$;
\item
for each $F \subseteq \Foci$, the unary \emph{encapsulation} operator
$\funct{\encap{F}}{\ServFam}{\ServFam}$.
\end{itemize}
We assume that there is a countably infinite set of variables of sort
$\ServFam$ which includes $u,v,w$.
Terms are built as usual in the many-sorted case
(see e.g.~\cite{Wir90a,ST99a}).
We use prefix notation for the singleton service family operators and
infix nota\-tion for the service family composition operator.

The service family denoted by $\emptysf$ is the empty service family.
The service family denoted by a closed term of the form $f.H$ consists
of one named service only, the service concerned is the service denoted
by $H$, and the name of this service is $f$.
The service family denoted by a closed term of the form $C \sfcomp D$
consists of all named services that belong to either the service family
denoted by $C$ or the service family denoted by $D$.
In the case where a named service from the service family denoted by $C$
and a named service from the service family denoted by $D$ have the same
name, they collapse to an empty service with the name concerned.
The service family denoted by a closed term of the form $\encap{F}(C)$
consists of all named services with a name not in $F$ that belong to the
service family denoted by $C$.
Thus, the service families denoted by closed terms of the forms $f.H$
and $\encap{\set{f}}(C)$ do not collapse to an empty service in service
family composition.

Using the singleton service family operators and the service family
composition operator, any finite number of possibly identical services
can be brought together in a service family provided that the services
concerned are given different names.
In Section~\ref{sect-example}, we will give an example of the use of
the singleton service family operators and the service family
composition operator.
The empty service family constant and the encapsulation operators are
primarily meant to axiomatize the operators that are introduced in
Section~\ref{sect-TSI}.

The service family composition operator takes the place of the
non-interfering combination operator from~\cite{BP02a}.
As suggested by the name, service family composition is composition of
service families.
Non-interfering combination is composition of services.
The non-interfering combination of services can process all methods
that can be processed by only one of the services.
This has the disadvantage that its usefulness is rather limited without
an additional renaming mechanism.
For example, a finite number of identical services cannot be brought
together in a service by means of non-interfering combination.

The axioms of \SFA\ are given in Table~\ref{axioms-SFA}.%
\begin{table}[!t]
\caption{Axioms of \SFA}
\label{axioms-SFA}
\begin{eqntbl}
\begin{axcol}
u \sfcomp \emptysf = u                                 & \axiom{SFC1} \\
u \sfcomp v = v \sfcomp u                              & \axiom{SFC2} \\
(u \sfcomp v) \sfcomp w = u \sfcomp (v \sfcomp w)      & \axiom{SFC3} \\
f.H \sfcomp f.H' = f.\emptyserv                        & \axiom{SFC4}
\end{axcol}
\qquad
\begin{saxcol}
\encap{F}(\emptysf) = \emptysf                       & & \axiom{SFE1} \\
\encap{F}(f.H) = \emptysf & \mif f \in F               & \axiom{SFE2} \\
\encap{F}(f.H) = f.H      & \mif f \notin F            & \axiom{SFE3} \\
\encap{F}(u \sfcomp v) =
\encap{F}(u) \sfcomp \encap{F}(v)                    & & \axiom{SFE4}
\end{saxcol}
\end{eqntbl}
\end{table}
In this table, $f$ stands for an arbitrary focus from $\Foci$ and $H$
and $H'$ stand for arbitrary closed terms of sort $\Serv$.
The axioms of \SFA\ simply formalize the informal explanation given above.

In Section~\ref{sect-example}, we will give an example of the use of the
service family composition operator.
There, we will write $\Sfcomp{i = 1}{n} C_i$, where $C_1,\ldots,C_n$ are
terms of sort $\ServFam$, for the term $C_1 \sfcomp \ldots \sfcomp C_n$.

\section{Use, Apply and Reply}
\label{sect-TSI}

A thread may interact with the named services from the service family
offered by an execution environment.
That is, a thread may perform a basic action for the purpose of
requesting a named service to process a method and to return a reply
value at completion of the processing of the method.
In this section, we combine \BTAbtp\ with \SFA\ and extend the
combination with three operators that relate to this kind of interaction
between threads and services, resulting in \TAbt.

The operators in question are called the use operator, the apply
operator, and the reply operator.
The difference between the use operator and the apply operator is a
matter of perspective: the use operator is concerned with the effects of
service families on threads and therefore produces threads, whereas the
apply operator is concerned with the effects of threads on service
families and therefore produces service families.
The reply operator is concerned with the effects of service families on
the Boolean values that threads possibly deliver at their termination.
The use operator and the apply operator introduced here are mainly
adaptations of the use operators and the apply operators introduced
in~\cite{BP02a} to service families.
The reply operator has no counterpart in~\cite{BP02a}.

The reply operator produces special values in the case where no Boolean
value is delivered at termination or no termination takes place.
Thus, it is accomplished that all terms with occurrences of the reply
operator denote something.
However, we prefer to use the reply operator only if it is known that
termination with delivery of a Boolean value takes place
(see also Section~\ref{sect-RUC}).

For the set $\BAct$ of basic actions, we take the set
$\set{f.m \where f \in \Foci, m \in \Meth}$.
All three operators mentioned above are concerned with the processing of
methods by services from a service family in pursuance of basic actions
performed by a thread.
The service involved in the processing of a method is the service whose
name is the focus of the basic action in question.

\TAbt\ has the sorts, constants and operators of both \BTAbtp\ and \SFA\
and in addition the following sort:
\begin{itemize}
\item
the sort $\Repl$ of \emph{replies};
\end{itemize}
and the following constants and operators:
\begin{itemize}
\item
the \emph{reply} constants $\const{\True,\False,\Div,\Mea}{\Repl}$;
\item
the binary \emph{use} operator
$\funct{\ph \sfuse \ph}{\Thr \x \ServFam}{\Thr}$;
\item
the binary \emph{apply} operator
$\funct{\ph \sfapply \ph}{\Thr \x \ServFam}{\ServFam}$;
\item
the binary \emph{reply} operator
$\funct{\ph \sfreply \ph}{\Thr \x \ServFam}{\Repl}$.
\end{itemize}
We use infix notation for the use, apply and reply operators.

The thread denoted by a closed term of the form $p \sfuse C$ and the
service family denoted by a closed term of the form $p \sfapply C$ are
the thread and service family, respectively, that result from processing
the method of each basic action performed by the thread denoted by $p$
by the service in the service family denoted by $C$ with the focus
of the basic action as its name if such a service exists.
When the method of a basic action performed by a thread is processed by
a service, the service changes in accordance with the method concerned,
and affects the thread as follows: the basic action turns into the
internal action $\Tau$ and the two ways to proceed reduce to one on the
basis of the reply value produced by the service.
The value denoted by a closed term of the form $p \sfreply C$ is the
Boolean value that the thread denoted by $p \sfuse C$ delivers at its
termination if it terminates and delivers a Boolean value at
termination, the value denoted by $\Mea$ (standing for meaningless) if
it terminates and does not deliver a Boolean value at termination, and
the value denoted by $\Div$ (standing for divergent) if it does not
terminate.
We are only interested in models of \TAbt\ in which the cardinality of
the domain associated with the sort $\Repl$ is $4$ and the elements of
this domain denoted by the constants $\True$, $\False$, $\Div$ and
$\Mea$ are mutually different.

A simple example of the application of the use operator, the apply
operator and the reply operator is
\begin{ldispl}
(\pcc{(\nns.\pop \bapf \StopP)}{\nns.\topeq{0}}{\StopN}) \sfuse
\nns.\NNS(0 \sigma)\;,
\\
(\pcc{(\nns.\pop \bapf \StopP)}{\nns.\topeq{0}}{\StopN}) \sfapply
\nns.\NNS(0 \sigma)\;,
\\
(\pcc{(\nns.\pop \bapf \StopP)}{\nns.\topeq{0}}{\StopN}) \sfreply
\nns.\NNS(0 \sigma)\;,
\end{ldispl}%
where $\NNS(\sigma)$ denotes a stack service as described in
Section~\ref{sect-SF} dealing with a stack whose content is represented
by the sequence $\sigma$.
The first term denotes the thread that performs $\Tau$ twice and then
terminates with delivery of the Boolean value $\True$.
The second term denotes the stack service dealing with a stack whose
content is $\sigma$, i.e.\ the content at termination of this thread,
and the third term denotes the reply value $\True$, i.e.\ the reply
value delivered at termination of this thread.

The axioms of \TAbt\ are the axioms of \BTAbtp, the axioms of \SFA, and
the axioms given in Tables~\ref{axioms-use}, \ref{axioms-apply}
and~\ref{axioms-reply}.%
\begin{table}[!t]
\caption{Axioms for the use operator}
\label{axioms-use}
\begin{eqntbl}
\begin{saxcol}
\StopP \sfuse u = \StopP                               & & \axiom{U1} \\
\StopN \sfuse u = \StopN                               & & \axiom{U2} \\
\Stop  \sfuse u = \Stop                                & & \axiom{U3} \\
\DeadEnd \sfuse u = \DeadEnd                           & & \axiom{U4} \\
(\Tau \bapf x) \sfuse u = \Tau \bapf (x \sfuse u)      & & \axiom{U5} \\
(\pcc{x}{f.m}{y}) \sfuse \encap{\set{f}}(u) =
\pcc{(x \sfuse \encap{\set{f}}(u))}{f.m}{(y \sfuse \encap{\set{f}}(u))}
                                                       & & \axiom{U6} \\
(\pcc{x}{f.m}{y}) \sfuse (f.H \sfcomp \encap{\set{f}}(u)) =
\Tau \bapf (x \sfuse (f.\derive{m}H \sfcomp \encap{\set{f}}(u)))
                           & \mif \sreply{m}(H) = \True  & \axiom{U7} \\
(\pcc{x}{f.m}{y}) \sfuse (f.H \sfcomp \encap{\set{f}}(u)) =
\Tau \bapf (y \sfuse (f.\derive{m}H \sfcomp \encap{\set{f}}(u)))
                           & \mif \sreply{m}(H) = \False & \axiom{U8} \\
(\pcc{x}{f.m}{y}) \sfuse (f.H \sfcomp \encap{\set{f}}(u)) = \DeadEnd
                           & \mif \sreply{m}(H) = \Div   & \axiom{U9} \\
\proj{n}{x \sfuse u} = \proj{n}{x} \sfuse u            & & \axiom{U10}
\end{saxcol}
\end{eqntbl}
\end{table}%
\begin{table}[!t]
\caption{Axioms for the apply operator}
\label{axioms-apply}
\begin{eqntbl}
\begin{saxcol}
\StopP \sfapply u = u                                  & & \axiom{A1} \\
\StopN \sfapply u = u                                  & & \axiom{A2} \\
\Stop  \sfapply u = u                                  & & \axiom{A3} \\
\DeadEnd \sfapply u = \emptysf                         & & \axiom{A4} \\
(\Tau \bapf x) \sfapply u = x \sfapply u               & & \axiom{A5} \\
(\pcc{x}{f.m}{y}) \sfapply \encap{\set{f}}(u) = \emptysf
                                                       & & \axiom{A6} \\
(\pcc{x}{f.m}{y}) \sfapply (f.H \sfcomp \encap{\set{f}}(u)) =
x \sfapply (f.\derive{m}H \sfcomp \encap{\set{f}}(u))
                           & \mif \sreply{m}(H) = \True  & \axiom{A7} \\
(\pcc{x}{f.m}{y}) \sfapply (f.H \sfcomp \encap{\set{f}}(u)) =
y \sfapply (f.\derive{m}H \sfcomp \encap{\set{f}}(u))
                           & \mif \sreply{m}(H) = \False & \axiom{A8} \\
(\pcc{x}{f.m}{y}) \sfapply (f.H \sfcomp \encap{\set{f}}(u)) = \emptysf
                           & \mif \sreply{m}(H) = \Div   & \axiom{A9} \\
\AND{n \geq 0}{} \proj{n}{x} \sfapply u = \proj{n}{y} \sfapply v
                 \Implies x \sfapply u = y  \sfapply v & & \axiom{A10}
\end{saxcol}
\end{eqntbl}
\end{table}%
\begin{table}[!t]
\caption{Axioms for the reply operator}
\label{axioms-reply}
\begin{eqntbl}
\begin{saxcol}
\StopP \sfreply u = \True                              & & \axiom{R1} \\
\StopN \sfreply u = \False                             & & \axiom{R2} \\
\Stop  \sfreply u = \Mea                               & & \axiom{R3} \\
\DeadEnd \sfreply u = \Div                             & & \axiom{R4} \\
(\Tau \bapf x) \sfreply u = x \sfreply u               & & \axiom{R5} \\
(\pcc{x}{f.m}{y}) \sfreply \encap{\set{f}}(u) = \Div   & & \axiom{R6} \\
(\pcc{x}{f.m}{y}) \sfreply (f.H \sfcomp \encap{\set{f}}(u)) =
x \sfreply (f.\derive{m}H \sfcomp \encap{\set{f}}(u))
                           & \mif \sreply{m}(H) = \True  & \axiom{R7} \\
(\pcc{x}{f.m}{y}) \sfreply (f.H \sfcomp \encap{\set{f}}(u)) =
y \sfreply (f.\derive{m}H \sfcomp \encap{\set{f}}(u))
                           & \mif \sreply{m}(H) = \False & \axiom{R8} \\
(\pcc{x}{f.m}{y}) \sfreply (f.H \sfcomp \encap{\set{f}}(u)) = \Div
                           & \mif \sreply{m}(H) = \Div   & \axiom{R9} \\
\AND{n \geq 0}{} \proj{n}{x} \sfreply u = \proj{n}{y} \sfreply v
                 \Implies x \sfreply u = y  \sfreply v & & \axiom{R10}
\end{saxcol}
\end{eqntbl}
\end{table}
In these tables, $f$ stands for an arbitrary focus from $\Foci$, $m$
stands for an arbitrary method from $\Meth$, $H$ stands for an arbitrary
term of sort $\Serv$, and $n$ stands for an arbitrary natural number.
The axioms simply formalize the informal explanation given above and in
addition stipulate what is the result of use, apply and reply if
inappropriate foci or methods are involved.
Axioms A10 and R10 allow for reasoning about infinite threads in the
contexts of apply and reply, respectively.
The counterpart of A10 and R10 for use, i.e.
\begin{ldispl}
\AND{n \geq 0}{} \proj{n}{x} \sfuse u = \proj{n}{y} \sfuse v
                \Implies x \sfuse u = y  \sfuse v\;,
\end{ldispl}%
follows from AIP and U10.

We can prove that each closed \TAbt\ term of sort $\Thr$ can be reduced
to a closed \BTAbt\ term of sort $\Thr$.
\begin{lemma}
\label{lemma-elimination}
For all closed \TAbt\ terms $p$ of sort $\Thr$, there exists a closed
\BTAbt\ term $q$ of sort $\Thr$ such that $p = q$ is derivable from the
axioms of \TAbt.
\end{lemma}
\begin{proof}
In the special case of singleton service families, this is in fact
proved in~\cite{BM06a} as part of the proof of Theorem~3 from that
paper.
The proof of the general case follows essentially the same lines.
\qed
\end{proof}

In the case of \TAbt, the notion of a guarded recursive specification
is somewhat adapted.
A \emph{guarded recursive specification} over \TAbt\ is a set of
recursion equations $E = \set{x = t_x \where x \in V}$, where $V$ is a
set of variables of sort $\Thr$ and each $t_x$ is a \TAbt\ term of sort
$\Thr$ that can be rewritten, using the axioms of \TAbt, to a term of
the form $\DeadEnd$, $\Stop$, $\StopP$, $\StopN$ or $\pcc{t}{a}{t'}$
with $t$ and $t'$ that contain only variables from $V$.
We are only interested in models of \TAbt\ in which guarded recursive
specifications have unique solutions.

A thread $p$ in a model $\gM$ of \TAbt\ in which guarded recursive
specifications over \TAbt\ have unique solutions is \emph{definable} if
it is the solution in $\gM$ of a guarded recursive specification over
\TAbt.
It is easy to see that a thread is definable if it is representable by
a closed \TAbt\ term of sort $\Thr$.

Henceforth, we assume that a model $\gM$ of \TAbt\ has been given in
which guarded recursive specifications over \TAbt\ have unique
solutions, all threads are definable, all service families are
representable by a closed \TAbt\ term of sort $\ServFam$, and all
replies are representable by a closed \TAbt\ term of sort $\Repl$.

Below, we will formulate a proposition about the use, apply and reply
operators using the \emph{foci} operation $\foci$ defined by the
equations in Table~\ref{eqns-foci}
(for foci $f \in \Foci$ and terms $H$ of sort $\Serv$).%
\begin{table}[!t]
\caption{Defining equations for the foci operation}
\label{eqns-foci}
\begin{eqntbl}
\begin{eqncol}
\foci(\emptysf) = \emptyset                                           \\
\foci(f.H) = \set{f}                                                  \\
\foci(u \sfcomp v) = \foci(u) \union \foci(v)
\end{eqncol}
\end{eqntbl}
\end{table}
The operation $\foci$ gives, for each service family, the set of all
foci that serve as names of named services belonging to the service
family.
We will make use of the following properties of $\foci$ in the proof of
the proposition:
\begin{enumerate}
\item
$\foci(u) \inter \foci(v) = \emptyset$ iff
$f \notin \foci(u)$ or $f \notin \foci(v)$ for all $f \in \Foci$;
\item
$f \not\in \foci(u)$ iff $\encap{\set{f}}(u) = u$.
\end{enumerate}

\begin{proposition}
\label{prop-ispo}
If $\foci(u) \inter \foci(v) = \emptyset$, then:
\begin{enumerate}
\item
$x \sfuse (u \sfcomp v) = (x \sfuse u) \sfuse v$;
\item
$x \sfreply (u \sfcomp v) = (x \sfuse u) \sfreply v$;
\item
$\encap{\foci(u)}(x \sfapply (u \sfcomp v)) = (x \sfuse u) \sfapply v$.
\end{enumerate}
\end{proposition}
\begin{proof}
By the definition of a guarded recursive specification over \TAbt,
Lemmas~\ref{lemma-projections-BTAbt} and~\ref{lemma-elimination}, and
axioms AIP, U10, A10 and R10, it is sufficient to prove for all closed
\BTAbt\ term $p$ of sort $\Thr$:
\begin{enumerate}
\item[]
$p \sfuse (u \sfcomp v) = (p \sfuse u) \sfuse v$;
\item[]
$p \sfreply (u \sfcomp v) = (p \sfuse u) \sfreply v$;
\item[]
$\encap{\foci(u)}(p \sfapply (u \sfcomp v)) = (p \sfuse u) \sfapply v$.
\end{enumerate}
This is straightforward by induction on the structure of $p$, using the
above-mentioned properties of $\foci$.
\qed
\end{proof}

Let $p$ and $C$ be \TAbt\ terms of sort $\Thr$ and $\ServFam$,
respectively.
Then $p$ \emph{converges on} $C$, written $p \cvg C$, is inductively
defined by the following clauses:
\begin{enumerate}
\item
$\Stop \cvg u$;
\item
$\StopP \cvg u$ and $\StopN \cvg u$;
\item
if $x \cvg u$, then $(\Tau \bapf x) \cvg u$;
\item
if $\sreply{m}(H) = \True$ and
$x \cvg (f.\derive{m}H \sfcomp \encap{\set{f}}(u))$, then
$(\pcc{x}{f.m}{y}) \cvg (f.H \sfcomp \encap{\set{f}}(u))$;
\item
if $\sreply{m}(H) = \False$ and
$y \cvg (f.\derive{m}H \sfcomp \encap{\set{f}}(u))$, then
$(\pcc{x}{f.m}{y}) \cvg (f.H \sfcomp \encap{\set{f}}(u))$;
\item
if $\proj{n}{x} \cvg u$, then $x \cvg u$;
\end{enumerate}
and $p$ \emph{diverges on} $C$, written $p \dvg C$, is defined by
$p \dvg C$ iff not $p \cvg C$.
Moreover, $p$ \emph{converges on} $C$ \emph{with Boolean reply}, written
$p \cvgb C$, is inductively defined by the clauses $2,\ldots,6$ for
$\cvg$ with everywhere $\cvg$ replaced by $\cvgb$.

The following two propositions concern the connection between
convergence and the reply operator.
\begin{proposition}
\label{prop-bt}
Let $p$ be a closed \TAbt\ term of sort $\Thr$.
Then:
\begin{enumerate}
\item
if $p \cvg u$, $\StopP$ occurs in $p$ and both $\StopN$ and $\Stop$ do
not occur in $p$, then $p \sfreply u = \True$;
\item
if $p \cvg u$, $\StopN$ occurs in $p$ and both $\StopP$ and $\Stop$ do
not occur in $p$, then $p \sfreply u = \False$;
\item
if $p \cvg u$, $\Stop$ occurs in $p$ and both $\StopP$ and $\StopN$ do
not occur in $p$, then $p \sfreply u = \Mea$.
\end{enumerate}
\end{proposition}
\begin{proof}
By Lemma~\ref{lemma-elimination}, it is sufficient to prove it for all
closed \BTAbt\ terms $p$ of sort $\Thr$.
This is straightforward by induction on the structure of $p$.
\qed
\end{proof}
\begin{proposition}
\label{prop-cvg-sfreply}
We have that
$x \cvg u$ iff
$x \sfreply u = \True$ or $x \sfreply u = \False$ or
$x \sfreply u = \Mea$.
\end{proposition}
\begin{proof}
By the definition of a guarded recursive specification over \TAbt, the
last clause of the inductive definition of $\cvg$,
Lemmas~\ref{lemma-projections-BTAbt} and~\ref{lemma-elimination}, and
axiom R10, it is sufficient to prove $p \cvg u$ iff
$p \sfreply u = \True$ or $p \sfreply u = \False$ or
$p \sfreply u = \Mea$ for all closed \BTAbt\ terms $p$ of sort $\Thr$.
This is straightforward by induction on the structure of $p$.
\qed
\end{proof}

Because the use operator, apply operator and reply operator are
primarily intended to be used to describe and analyse instruction
sequence processing, they are called \emph{instruction sequence
processing operators}.

We introduce the apply operator and reply operator in the setting of
\PGLBbt\ by defining:
\begin{ldispl}
P \sfuse   u = \extr{P} \sfuse   u\;, \quad
P \sfapply u = \extr{P} \sfapply u\;, \quad
P \sfreply u = \extr{P} \sfreply u
\end{ldispl}%
for all \PGLBbt\ instruction sequences $P$.
Similarly, we introduce convergence in the setting of \PGLBbt\ by
defining:
\begin{ldispl}
P \cvg u = \extr{P} \cvg u
\end{ldispl}%
for all \PGLBbt\ instruction sequences $P$.

\section{Relevant Use Conventions}
\label{sect-RUC}

In the setting of service families, sets of foci play the role of
interfaces.
The set of all foci that serve as names of named services in a service
family is regarded as the interface of that service family.
There are cases in which processing does not terminate or, even worse
(because it is statically detectable), interfaces of services families
do not match.
In the case of non-termination, there is nothing that we intend to
denote by a term of the form $p \sfapply C$ or $p \sfreply C$.
In the case of non-matching services families, there is nothing that we
intend to denote by a term of the form $C \sfcomp D$.
Moreover, in the case of termination without a Boolean reply, there is
nothing that we intend to denote by a term of the form $p \sfreply C$.

We propose to comply with the following \emph{relevant use conventions}:
\begin{itemize}
\item
$p \sfapply C$ is only used if it is known that $p \cvg C$;
\item
$p \sfreply C$ is only used if it is known that $p \cvgb C$;
\item
$C \sfcomp D$ is only used if it is known that
$\foci(C) \inter \foci(D) = \emptyset$.
\end{itemize}

The condition found in the first convention is justified by the fact
that $x \sfapply u = \emptysf$ if $x \dvg u$.
We do not have $x \sfapply u = \emptysf$ only if $x \dvg u$.
For instance, $\StopP \sfapply \emptysf = \emptysf$ whereas
$\StopP \cvg \emptysf$.
Similar remarks apply to the condition found in the second convention.

The idea of relevant use conventions is taken from~\cite{BM09g}, where
it plays a central role in an account of the way in which mathematicians
usually deal with division by zero in mathematical texts.
According to~\cite{BM09g}, mathematicians deal with this issue by
complying with the convention that $p \mdiv q$ is only used if it is
known that $q \neq 0$.
This approach is justified by the fact that there is nothing that
mathematicians intend to denote by $p \mdiv q$ if $q = 0$.
It yields simpler mathematical texts than the popular approach in
theoretical computer science, which is characterized by complete
formality in definitions, statements and proofs.
In this computer science approach, division is considered a partial
function and some logic of partial functions is used.
In~\cite{BT07a}, deviating from this, division is considered a total
function whose value is zero in all cases of division by zero.
It may be imagined that this notion of division is the one with which
mathematicians make themselves familiar before they start to read and
write mathematical texts professionally.

We think that the idea to comply with conventions that exclude the use
of terms that are not really intended to denote anything is not only of
importance in mathematics, but also in theoretical computer science.
For example, the consequence of adapting Proposition~\ref{prop-ispo} to
comply with the relevant use conventions described above, by adding
appropriate conditions to the three properties, is that we do not have
to consider in the proof of the proposition the equality of terms by
which we do not intend to denote anything.

In the sequel, we will comply with the relevant use conventions
described above.

We can define the use operators introduced earlier
in~\cite{BM07g,BM06c},%
\footnote
{The use operators introduced in~\cite{BP02a} are counterparts of the
 abstracting use operator introduced later in
 Section~\ref{sect-abstr-use}.}
the apply operators introduced earlier in~\cite{BP02a}, and similar
counterparts of the reply operator as follows:
\begin{ldispl}
\use{x}{f}{H} = x \sfuse f.H\;, \\
\apply{x}{f}{H} = x \sfapply f.H\;, \\
\reply{x}{f}{H} = x \sfreply f.H\;.
\end{ldispl}%
These definitions give rise to the derived conventions that
$\apply{p}{f}{H}$ is only used if it is known that $p \cvg f.H$ and
$\reply{p}{f}{H}$ is only used if it is known that $p \cvgb f.H$.

\section{Example}
\label{sect-example}

In this section, we use an implementation of a bounded counter by means
of a number of Boolean registers as an example to show that it is easy
to bring a number of identical services together in a service family by
means of the service family composition operator.
Accomplishing something resemblant with the non-interfering service
combination operation from~\cite{BP02a} is quite involved.
We also show in this example that there are cases in which the delivery
of a Boolean value at termination of the execution of an instruction
sequence is quite natural.

First, we describe services that make up Boolean registers.
The Boolean register services are able to process the following methods:
\begin{itemize}
\item
the \emph{set to true method} $\setbr{\True}$;
\item
the \emph{set to false method} $\setbr{\False}$;
\item
the \emph{get method} $\getbr$.
\end{itemize}
It is assumed that $\setbr{\True},\setbr{\False},\getbr \in \Meth$.

The methods that Boolean register services are able to process can be
explained as follows:
\begin{itemize}
\item
$\setbr{\True}$\,:
the contents of the Boolean register becomes $\True$ and the reply is
$\True$;
\item
$\setbr{\False}$\,:
the contents of the Boolean register becomes $\False$ and the reply is
$\False$;
\item
$\getbr$\,:
nothing changes and the reply is the contents of the Boolean register.
\end{itemize}

For the set $\Services$ of services, we take the set
$\set{\BR{\True},\BR{\False},\BR{\Div}}$ of \emph{Boolean register
services}.
For each $m \in \Meth$, we take the functions $\effect{m}$ and
$\sreply{m}$ such that ($b \in \set{\True,\False}$):
\begin{ldispl}
\begin{geqns}
\effect{\setbr{\True}}(\BR{b})  = \BR{\True}\;,
\\
\effect{\setbr{\False}}(\BR{b}) = \BR{\False}\;,
\\
\effect{\getbr}(\BR{b}) = \BR{b}\;,
\eqnsep
\sreply{\setbr{\True}}(\BR{b})  = \True\;,
\\
\sreply{\setbr{\False}}(\BR{b}) = \False\;,
\\
\sreply{\getbr}(\BR{b}) = b\;,
\end{geqns}
\qquad\qquad
\begin{gceqns}
\effect{m}(\BR{b}) = \BR{\Div}
 & \mif m \not\in \set{\setbr{\True},\setbr{\False},\getbr}\;,
\\
\effect{m}(\BR{\Div}) = \BR{\Div}\;,
\\
{}
\eqnsep
\sreply{m}(\BR{b}) = \Div
 & \mif m \not\in \set{\setbr{\True},\setbr{\False},\getbr}\;,
\\
\sreply{m}(\BR{\Div}) = \Div\;.
\end{gceqns}
\end{ldispl}%
Moreover, we take the names used above to denote the services in
$\Services$ for constants of sort $\Serv$.

We continue with the implementation of a bounded counter by means of a
number of Boolean registers.
We consider a counter that can contain a natural number in the interval
$[0,2^n - 1]$ for some $n > 0$.
To implement the counter, we represent its content in binary using a
collection of $n$ Boolean registers named $\br{0},\ldots,\br{n{-}1}$.
We take $\True$ for $0$ and $\False$ for $1$, and we take the bit
represented by the content of the Boolean register named $\br{i}$ for a
less significant bit than the bit represented by the content of the
Boolean register named $\br{j}$ if $i < j$.

The following instruction sequences implement set to zero, increment by
one, decrement by one, and test on zero, respectively:
\begin{ldispl}
\begin{aeqns}
\nm{SETZERO} & = &
\Conc{i=0}{n-1} (\br{i}.\setbr{\True}) \conc \haltP\;,
\eqnsep
\nm{SUCC}    & = &
\Conc{i=0}{n-1}
 (\ntst{\br{i}.\getbr} \conc \fjmp{3} \conc \br{i}.\setbr{\False} \conc
  \haltP \conc \br{i}.\setbr{\True}) \conc \haltN\;,
\eqnsep
\nm{PRED}    & = &
\Conc{i=0}{n-1}
 (\ptst{\br{i}.\getbr} \conc \fjmp{3} \conc \br{i}.\setbr{\True} \conc
  \haltP \conc \br{i}.\setbr{\False}) \conc \haltN\;,
\eqnsep
\nm{ISZERO}  & = &
\Conc{i=0}{n-1}
 (\ntst{\br{i}.\getbr} \conc \haltN) \conc \haltP\;.
\end{aeqns}
\end{ldispl}%
Concerning the Boolean values delivered at termination of executions of
these instruction sequences, we have that:
\begin{ldispl}
\begin{aeqns}
\nm{SETZERO} \sfreply \bigl(\Sfcomp{i=0}{n-1} \br{i}.\BR{s_i}\bigr)
 & = & \True\;,
\beqnsep
\nm{SUCC} \sfreply \bigl(\Sfcomp{i=0}{n-1} \br{i}.\BR{s_i}\bigr)
 & = &
\Biggl\{
\begin{array}[c]{@{}l@{\;\;}l@{}}
\True  & \mif \OR{i=0}{n-1} s_i  = \True
\\
\False & \mif \AND{i=0}{n-1} s_i = \False\;,
\end{array}
\beqnsep
\nm{PRED} \sfreply \bigl(\Sfcomp{i=0}{n-1} \br{i}.\BR{s_i}\bigr)
 & = &
\Biggl\{
\begin{array}[c]{@{}l@{\;\;}l@{}}
\True  & \mif \OR{i=0}{n-1} s_i = \False
\\
\False & \mif \AND{i=0}{n-1} s_i  = \True\;,
\end{array}
\beqnsep
\nm{ISZERO} \sfreply \bigl(\Sfcomp{i=0}{n-1} \br{i}.\BR{s_i}\bigr)
 & = &
\Biggl\{
\begin{array}[c]{@{}l@{\;\;}l@{}}
\True  & \mif \AND{i=0}{n-1} s_i = \True
\\
\False & \mif \OR{i=0}{n-1} s_i  = \False\;.
\end{array}
\end{aeqns}
\end{ldispl}%
It is obvious that $\True$ is delivered at termination of an execution
of $\nm{SETZERO}$ and that $\True$ or $\False$ is delivered at
termination of an execution of $\nm{ISZERO}$ depending on whether the
content of the counter is zero or not.
Increment by one and decrement by one are both modulo $2^n$.
For that reason, $\True$ or $\False$ is delivered at termination of an
execution of $\nm{SUCC}$ or $\nm{PRED}$ depending on whether the
content of the counter is really incremented or decremented by one or
not.

\section{Abstracting Use}
\label{sect-abstr-use}

With the use operator introduced in Section~\ref{sect-TSI}, the action
$\Tau$ is left as a trace of a basic action that has led to the
processing of a method, like with the use operators on services
introduced in e.g.~\cite{BM07g,BM06c}.
However, with the use operators on services introduced in~\cite{BP02a},
nothing is left as a trace of a basic action that has led to the
processing of a method.
Thus, these use operators abstract fully from internal activity.
In other words, they are abstracting use operators.
For completeness, we introduce an abstracting variant of the use
operator introduced in Section~\ref{sect-TSI}.

That is, we introduce the following additional operator:
\begin{itemize}
\item
the binary \emph{abstracting use} operator
$\funct{\ph \sfause \ph}{\Thr \x \ServFam}{\Thr}$.
\end{itemize}
We use infix notation for the abstracting use operator.

The axioms for the abstracting use operator are given in
Table~\ref{axioms-ause}.%
\begin{table}[!t]
\caption{Axioms for the abstracting use operator}
\label{axioms-ause}
\begin{eqntbl}
\begin{saxcol}
\StopP \sfause u = \StopP                             & & \axiom{AU1} \\
\StopN \sfause u = \StopN                             & & \axiom{AU2} \\
\Stop  \sfause u = \Stop                              & & \axiom{AU3} \\
\DeadEnd \sfause u = \DeadEnd                         & & \axiom{AU4} \\
(\Tau \bapf x) \sfause u = \Tau \bapf (x \sfause u)   & & \axiom{AU5} \\
(\pcc{x}{f.m}{y}) \sfause \encap{\set{f}}(u) =
\pcc{(x \sfause \encap{\set{f}}(u))}
 {f.m}{(y \sfause \encap{\set{f}}(u))}                & & \axiom{AU6} \\
(\pcc{x}{f.m}{y}) \sfause (f.H \sfcomp \encap{\set{f}}(u)) =
x \sfause (f.\derive{m}H \sfcomp \encap{\set{f}}(u))
                          & \mif \sreply{m}(H) = \True  & \axiom{AU7} \\
(\pcc{x}{f.m}{y}) \sfause (f.H \sfcomp \encap{\set{f}}(u)) =
y \sfause (f.\derive{m}H \sfcomp \encap{\set{f}}(u))
                          & \mif \sreply{m}(H) = \False & \axiom{AU8} \\
(\pcc{x}{f.m}{y}) \sfause (f.H \sfcomp \encap{\set{f}}(u)) = \DeadEnd
                          & \mif \sreply{m}(H) = \Div   & \axiom{AU9} \\
\AND{n \geq 0}{} \proj{n}{x} \sfause u = \proj{n}{y} \sfause v
                   \Implies x \sfause u = y \sfause v & & \axiom{AU10}
\end{saxcol}
\end{eqntbl}
\end{table}
Owing to the possible concealment of actions by abstracting use,
$\proj{n}{x \sfause u} = \proj{n}{x} \sfause u$ is not a plausible
axiom.
However, axiom AU10 allows for reasoning about infinite threads in the
context of abstracting use.

\section{Functional Units}
\label{sect-func-unit}

In this section, we introduce the concept of a functional unit and
related concepts.

It is assumed that a non-empty finite or countably infinite set $\FUS$
of \emph{states} has been given.
As before, it is assumed that a non-empty finite set $\MN$ of methods
has been given.
However, in the setting of functional units, methods serve as names of
operations on a state space.
For that reason, the members of $\MN$ will henceforth be called
\emph{method names}.

A \emph{method operation} on $\FUS$ is a total function from $\FUS$ to
$\Bool \x \FUS$.
A \emph{partial method operation} on $\FUS$ is a partial function from
$\FUS$ to $\Bool \x \FUS$.
We write $\MO(\FUS)$ for the set of all method operations on $\FUS$.
We write $M^r$ and $M^e$, where $M \in \MO(\FUS)$, for the unique
functions $\funct{R}{\FUS}{\Bool}$ and $\funct{E}{\FUS}{\FUS}$,
respectively, such that $M(s) = \tup{R(s),E(s)}$ for all $s \in \FUS$.

A \emph{functional unit} for $\FUS$ is a finite subset $\cH$ of
$\MN \x \MO(\FUS)$ such that $\tup{m,M} \in \cH$ and
$\tup{m,M'} \in \cH$ implies $M = M'$.
We write $\FU(\FUS)$ for the set of all functional units for $\FUS$.
We write $\IF(\cH)$, where $\cH \in \FU(\FUS)$, for the set
$\set{m \in \MN \where \Exists{M \in \MO(\FUS)}{\tup{m,M} \in \cH}}$.
We write $m_\cH$, where $\cH \in \FU(\FUS)$ and $m \in \IF(\cH)$, for
the unique $M \in \MO(\FUS)$ such that $\tup{m,M} \in \cH$.

We look upon the set $\IF(\cH)$, where $\cH \in \FU(\FUS)$, as the
interface of $\cH$.
It looks to be convenient to have a notation for the restriction of a
functional unit to a subset of its interface.
We write $\tup{I,\cH}$, where $\cH \in \FU(\FUS)$ and
$I \subseteq \IF(\cH)$, for the functional unit
$\set{\tup{m,M} \in \cH \where m \in I}$.

Let $\cH \in \FU(\FUS)$.
Then an \emph{extension} of $\cH$ is an $\cH' \in \FU(\FUS)$ such that
$\cH \subseteq \cH'$.

The following is a simple illustration of the use of functional units.
An unbounded counter can be modelled by a functional unit for $\Nat$
with method operations for set to zero, increment by one, decrement by
one, and test on zero.

According to the definition of a functional unit,
$\emptyset \in \FU(\FUS)$.
By that we have a unique functional unit with an empty interface, which
is not very interesting in itself.
However, when considering services that behave according to functional
units, $\emptyset$ is exactly the functional unit according to which the
empty service $\emptyserv$ (the service that is not able to process any
method) behaves.

The method names attached to method operations in functional units
should not be confused with the names used to denote specific method
operations in describing functional units.
Therefore, we will comply with the convention to use names beginning
with a lower-case letter in the former case and names beginning with an
upper-case letter in the latter case.

We will use \PGLBsbt\ instruction sequences to derive partial method
operations from the method operations of a functional unit.
We write $\Lf{I}$, where $I \subseteq \MN$, for the set of all \PGLBsbt\
instruction sequences, taking the set $\set{f.m \where m \in I}$ as the
set $\BInstr$ of basic instructions.

The derivation of partial method operations from the method operations
of a functional unit involves services whose processing of methods
amounts to replies and service changes according to corresponding method
operations of the functional unit concerned.
These services can be viewed as the behaviours of a machine, on which
the processing in question takes place, in its different states.
We take the set $\FU(\FUS) \x \FUS$ as the set $\Services$ of services.
We write $\cH(s)$, where $\cH \in \FU(\FUS)$ and $s \in \FUS$, for the
service $\tup{\cH,s}$.
The functions $\effect{m}$ and $\sreply{m}$ are defined as follows:
\pagebreak[2]
\begin{ldispl}
\begin{aeqns}
\effect{m}(\cH(s)) & = &
\Biggl\{
\begin{array}[c]{@{}l@{\;\;}l@{}}
\cH(m_\cH^e(s))            & \mif m \in \IF(\cH) \\
{\emptyset}(s')            & \mif m \notin \IF(\cH)\;,
\end{array}
\beqnsep
\sreply{m}(\cH(s))  & = &
\Biggl\{
\begin{array}[c]{@{}l@{\;\;}l@{}}
m_\cH^r(s) \phantom{\cH()} & \mif m \in \IF(\cH) \\
\Div                       & \mif m \notin \IF(\cH)\;,
\end{array}
\end{aeqns}
\end{ldispl}%
where $s'$ is a fixed but arbitrary state in $S$.
In order to be able to make use of the axioms for the apply operator
and the reply operator from Section~\ref{sect-TSI} hereafter, we want
to use these operators for the services being considered here when
making the idea of deriving a partial method operation by means of an
instruction sequence precise.
Therefore, we assume that there is a constant of sort $\Serv$ for each
$\cH(s) \in \Services$.%
\footnote
{This may lead to an uncountable number of constants, which is
 unproblematic and quite normal in model theory.
}
In this connection, we use the following notational convention: for
each $\cH(s) \in \Services$, we write $\cterm{\cH(s)}$ for the constant
of sort $\Serv$ whose interpretation is $\cH(s)$.
Note that the service ${\emptyset}(\sigma')$ is the interpretation of
the empty service constant $\emptyserv$.

Let $\cH \in \FU(\FUS)$, and let $I \subseteq \IF(\cH)$.
Then an instruction sequence $x \in \Lf{I}$ produces a partial method
operation $\moextr{x}{\cH}$ as follows:
\begin{ldispl}
\begin{aceqns}
\moextr{x}{\cH}(s) & = &
\tup{\moextr{x}{\cH}^r(s),\moextr{x}{\cH}^e(s)}
 & \mif \moextr{x}{\cH}^r(s) = \True \Or
        \moextr{x}{\cH}^r(s) = \False\;, \\
\moextr{x}{\cH}(s) & \mathrm{is} & \mathrm{undefined}
 & \mif \moextr{x}{\cH}^r(s) = \Div\;,
\end{aceqns}
\end{ldispl}%
where
\begin{ldispl}
\begin{aeqns}
\moextr{x}{\cH}^r(s) & = & x \sfreply f.\cterm{\cH(s)}\;, \\
\moextr{x}{\cH}^e(s) & = &
\mathrm{the\;unique}\; s' \in S\; \mathrm{such\;that}\;
 x \sfapply f.\cterm{\cH(s)} = f.\cterm{\cH(s')}\;.
\end{aeqns}
\end{ldispl}%
If $\moextr{x}{\cH}$ is total, then it is called a
\emph{derived method operation} of $\cH$.

The binary relation $\below$ on $\FU(\FUS)$ is defined by
$\cH \below \cH'$ iff for all $\tup{m,M} \in \cH$, $M$ is a derived
method operation of $\cH'$.
The binary relation $\equiv$ on $\FU(\FUS)$ is defined by
$\cH \equiv \cH'$ iff $\cH \below \cH'$ and $\cH' \below \cH$.

\begin{theorem}
\label{theorem-below-equiv}
\mbox{}
\begin{enumerate}
\item
$\below$ is transitive;
\item
$\equiv$ is an equivalence relation.
\end{enumerate}
\end{theorem}
\begin{proof}
Property~1:
We have to prove that $\cH \below \cH'$ and $\cH' \below \cH''$ implies
$\cH \below \cH''$.
It is sufficient to show that we can obtain instruction sequences in
$\Lf{\IF(\cH'')}$ that produce the method operations of $\cH$ from the
instruction sequences in $\Lf{\IF(\cH')}$ that produce the method
operations of $\cH$ and the instruction sequences in $\Lf{\IF(\cH'')}$
that produce the method operations of $\cH'$.
Without loss of generality, we may assume that all instruction sequences
are of the form $u_1 \conc \ldots \conc u_k \conc \haltP \conc \haltN$,
where, for each $i \in [1,k]$, $u_i$ is a positive test instruction, a
forward jump instruction or a backward jump instruction.
Let $m \in \IF(\cH)$,
let $M$ be such that $\tup{m,M} \in \cH$, and
let $x_m \in \Lf{\IF(\cH')}$ be such that $M = \moextr{x_m}{\cH'}$.
Suppose that $\IF(\cH') = \set{m'_1,\ldots,m'_n}$.
For each $i \in [1,n]$,
let $M'_i$ be such that $\tup{m'_i,M'_i} \in \cH'$ and
let
$x_{m'_i} = u_1^i \conc \ldots \conc u_{k_i}^i \conc \haltP \conc \haltN
  \in \Lf{\IF(\cH'')}$ be such that $M'_i = \moextr{x_{m'_i}}{\cH''}$.
Consider the $x'_m \in \Lf{\IF(\cH'')}$ obtained from $x_m$ as follows:
for each $i \in [1,n]$,
(i)~first increase each jump over the leftmost occurrence of
$\ptst{f.m'_i}$ in $x_m$ with $k_i + 1$, and next replace this
instruction by $u_1^i \conc \ldots \conc u_{k_i}^i$;
(ii)~repeat the previous step as long as their are occurrences of
$\ptst{f.m'_i}$.
It is easy to see that $M = \moextr{x'_m}{\cH''}$.

Property~2:
It follows immediately from the definition of $\equiv$ that $\equiv$ is
symmetric and from the definition of $\below$ that $\below$ is
reflexive.
From these properties, Property~1 and the definition of $\equiv$, it
follows immediately that $\equiv$ is symmetric, reflexive and
transitive.
\qed
\end{proof}

The members of the quotient set $\FU(\FUS) \mdiv {\equiv}$ are called
\emph{functional unit degrees}.
Let $\cH \in \FU(\FUS)$ and $\cD \in \FU(\FUS) \mdiv {\equiv}$.
Then $\cD$ is a \emph{functional unit degree below}~$\cH$ if there
exists an $\cH' \in \cD$ such that $\cH' \below \cH$.

Two functional units $\cH$ and $\cH'$ belong to the same functional
unit degree if and only if $\cH$ and $\cH'$ have the same derived
method operations.
A functional unit degree $\cD$ is below a functional unit $\cH$ if and
only if all derived method operations of some member of $\cD$ are
derived method operations of $\cH$.

The binary relation $\below$ on $\FU(\FUS)$ is reminiscent of the
relative computability relation $\below$ on algebras introduced
in~\cite{LB81a} because functional units can be looked upon as
algebras of a special kind.
In the definition of this relative computability relation on algebras,
the role of instruction sequences is filled by flow charts.
A more striking difference is that the relation allows for algebras
with different domains to be related.
This corresponds to a relation on functional units that allows for the
states from one state space to be represented by the states from
another state space.
To the best of our knowledge, the work presented in~\cite{LB81a} and a
few preceding papers of the same authors is the only work on
computability that is concerned with a relation comparable to the
relation $\below$ on $\FU(\FUS)$ defined above.

\section{Functional Units for Natural Numbers}
\label{sect-func-unit-nat}

In this section, we investigate functional units for natural numbers.
The main consequences of considering the special case where the state
space is $\Nat$ are the following:
(i)~$\Nat$ is infinite,
(ii)~there is a notion of computability known which can be used without
further preparations.

An example of a functional unit in $\FU(\Nat)$ is an unbounded counter.
The method names involved are $\setzero$, $\incr$, $\decr$, and
$\iszero$.
The method operations involved are the functions $\Setzero$, $\Incr$,
$\Decr$, $\funct{\Iszero}{\Nat}{\Bool \x \Nat}$ defined as follows:
\begin{ldispl}
\begin{aeqns}
\Setzero(x) & = & \tup{\True,0}\;,     \\
\Incr(x)    & = & \tup{\True,x + 1}\;, \\
\Decr(x)    & = &
\Biggl\{
\begin{array}[c]{@{}l@{\;\;}l@{}}
\tup{\True,x - 1} & \mif x > 0\;, \\
\tup{\False,0}    & \mif x = 0\;,
\end{array}
\beqnsep
\Iszero(x) & = &
\Biggl\{
\begin{array}[c]{@{}l@{\;\;}l@{}}
\tup{\True,x} \!\phantom{{} - 1} & \mif x = 0\;, \\
\tup{\False,x}                   & \mif x > 0\;.
\end{array}
\end{aeqns}
\end{ldispl}%
The functional unit $\Counter$ is defined as follows:
\begin{ldispl}
\Counter =
\set{\tup{\setzero,\Setzero},\tup{\incr,\Incr},
     \tup{\decr,\Decr},\tup{\iszero,\Iszero}}\;.
\end{ldispl}%
The following proposition shows that there are infinitely many
functional units for natural numbers with mutually different sets of
derived method operations whose method operations are derived method
operations of a major restriction of the functional unit $\Counter$.
\begin{proposition}
\label{prop-inf-many-degrees}
\sloppy
There are infinitely many functional unit degrees below
$\tup{\set{\decr,\iszero},\linebreak[2]\Counter}$.
\end{proposition}
\begin{proof}
For each $n \in \Nat$, we define a functional unit $\cH_n \in \FU(\Nat)$
such that $\cH_n \leq \tup{\set{\decr,\iszero},\Counter}$ as follows:
\begin{ldispl}
\cH_n = \set{\tup{\decrn{n},\Decrn{n}},\tup{\iszero,\Iszero}}\;,
\end{ldispl}%
where
\begin{ldispl}
\Decrn{n}(x) =
\biggl\{
\begin{array}[c]{@{}l@{\;\;}l@{}}
\tup{\True, x - n} & \mif x \geq n \\
\tup{\False,0}     & \mif x < n\;.
\end{array}
\end{ldispl}%
Let $n,m \in \Nat$ be such that $n < m$.
Then $\Decrn{n}(m) = \tup{\True,m - n}$.
However, there does not exist an $x \in \Lf{\IF(\cH_m)}$ such that
$\moextr{x}{\cH_m}(m) = \tup{\True,m - n}$ because
$\Decrn{m}(m) = \tup{\True,0}$.
Hence, $\cH_n \not\leq \cH_m$ for all $n,m \in \Nat$ with $n < m$.
\qed
\end{proof}

A method operation $M \in \MO(\Nat)$ is \emph{computable} if there exist
computable functions $\funct{F,G}{\Nat}{\Nat}$ such that
$M(n) = \tup{\beta(F(n)),G(n)}$ for all $n \in \Nat$,
where $\funct{\beta}{\Nat}{\Bool}$ is inductively defined by
$\beta(0) = \True$ and $\beta(n + 1) = \False$.
A functional unit $\cH \in \FU(\Nat)$ is \emph{computable} if, for each
$\tup{m,M} \in \cH$, $M$ is computable.

\begin{theorem}
\label{theorem-computable}
Let $\cH,\cH' \in \FU(\Nat)$ be such that $\cH \below \cH'$.
Then $\cH$ is computable if $\cH'$ is computable.
\end{theorem}
\begin{proof}
We will show that all derived method operations of $\cH'$ are
computable.

Take an arbitrary $P \in \Lf{\IF(\cH')}$ such that $\moextr{P}{\cH'}$ is
a derived method operations of $\cH'$.
It follows immediately from the definition of thread extraction that
$\extr{P}$ is the solution of a finite linear recursive specification
over \BTAbt, i.e.\ a finite guarded recursive specification over \BTAbt\
in which the right-hand side of each equation is a \BTAbt\ term of the
form $\DeadEnd$, $\StopP$, $\StopN$ or $\pcc{x}{a}{y}$ where $x$ and $y$
are variables of sort $\Thr$.
Let $E$ be a finite linear recursive specification over \BTAbt\ of which
the solution for $x_1$ is $\extr{P}$.
Because $\moextr{P}{\cH'}$ is total, it may be assumed without loss of
generality that $\DeadEnd$ does not occur as the right-hand side of an
equation in $E$.
Suppose that
\begin{ldispl}
E =
\set{x_i = \pcc{x_{l(i)}}{f.m_i}{x_{r(i)}} \where i \in [1,n]} \union
\set{x_{n+1} = \StopP, x_{n+2} = \StopN}\;.
\end{ldispl}%
From this set of equations, using the relevant axioms and definitions,
we obtain a set of equations of which the solution for $F_1$ is
$\moextr{P}{\cH'}^e$:
\begin{ldispl}
\set{F_i(s) = F_{l(i)}({m_i}_{\cH'}^e(s)) \mmul \nsg(\chi_i(s)) +
              F_{r(i)}({m_i}_{\cH'}^e(s)) \mmul  \sg(\chi_i(s))
      \where i \in [1,n]}
\\ \quad {} \union
\set{F_{n+1}(s) = s, F_{n+2}(s) = s}\;,
\end{ldispl}%
where, for every $i \in [1,n]$, the function
$\funct{\chi_i}{\Nat}{\Nat}$ is such that for all $s \in \Nat$:
\begin{ldispl}
\chi_i(s) = 0 \;\Iff\; {m_i}_{\cH'}^r(s) = \True\;,
\end{ldispl}%
and the functions $\funct{\sg,\nsg}{\Nat}{\Nat}$ are defined as usual:
\begin{ldispl}
\begin{aeqns}
\sg(0)     & = & 0\;, \\
\sg(n + 1) & = & 1\;,
\end{aeqns}
\qquad\qquad
\begin{aeqns}
\nsg(0)     & = & 1\;, \\
\nsg(n + 1) & = & 0\;.
\end{aeqns}
\end{ldispl}%
It follows from the way in which this set of equations is obtained from
$E$, the fact that ${m_i}_{\cH'}^e$ and $\chi_i$ are computable for each
$i \in [1,n]$, and the fact that $\sg$ and $\nsg$ are computable, that
this set of equations is equivalent to a set of equations by which
$\moextr{P}{\cH'}^e$ is defined recursively in the sense
of~\cite{Kle36a}.
This means that $\moextr{P}{\cH'}^e$ is general recursive, and hence
computable.

In a similar way, it is proved that $\moextr{P}{\cH'}^r$ is computable.
\qed
\end{proof}

A computable $\cH \in \FU(\Nat)$ is \emph{universal} if for each
computable $\cH' \in \FU(\Nat)$, we have $\cH' \below \cH$.
There exists a universal computable functional unit for natural numbers.
\begin{theorem}
\label{theorem-universal-fu}
There exists a computable $\cH \in \FU(\Nat)$ that is universal.
\end{theorem}
\begin{proof}
We will show that there exists a computable $\cH \in \FU(\Nat)$ with the
property that each computable $M \in \MO(\Nat)$ is a derived method
operation of $\cH$.

As a corollary of Theorem~10.3 from~\cite{SS63a},%
\footnote{That theorem can be looked upon as a corollary of Theorem~Ia
          from~\cite{Min61a}.}
we have that each computable $M \in \MO(\Nat)$ can be computed by means
of a register machine with six registers, say $\rmreg{0}$, $\rmreg{1}$,
$\rmreg{2}$, $\rmreg{3}$, $\rmreg{4}$, and $\rmreg{5}$.
The registers are used as follows:
$\rmreg{0}$ as input register;
$\rmreg{1}$ as output register for the output in $\Bool$;
$\rmreg{2}$ as output register for the output in $\Nat$;
$\rmreg{3}$, $\rmreg{4}$ and $\rmreg{5}$ as auxiliary registers.
The content of $\rmreg{1}$ represents the Boolean output as follows:
$0$ represents $\True$ and all other natural numbers represent $\False$.
For each $i \in [0,5]$, register $\rmreg{i}$ can be
incremented by one, decremented by one, and tested for zero by means of
instructions
$\rmreg{i}.\rmincr$, $\rmreg{i}.\rmdecr$ and $\rmreg{i}.\rmiszero$,
respectively.
We write $\RML$ for the set of all \PGLBsbt\ instruction sequences,
taking the set
$\set{\rmreg{i}.\rmincr,\rmreg{i}.\rmdecr,\rmreg{i}.\rmiszero \where
      i \in [0,5]}$
as the set $\BInstr$ of basic instructions.
Clearly, $\RML$ is adequate to represent all register machine programs
using six registers.

We define a computable functional unit $\Univ \in \FU(\Nat)$ whose
method operations can simulate the effects of the register machine
instructions by encoding the register machine states by natural numbers
such that the contents of the registers can reconstructed by prime
factorization.
This functional unit is defined as follows:
\begin{ldispl}
\begin{aeqns}
\Univ & = &
\set{\tup{\expii,\Expii},\tup{\factv,\Factv}}
\\    & \union &
\set{\tup{\rmmn{i}{succ},\rmmo{i}{succ}},
     \tup{\rmmn{i}{pred},\rmmo{i}{pred}},
     \tup{\rmmn{i}{iszero},\rmmo{i}{iszero}} \where
     i \in [0,5]}\,,
\end{aeqns}
\end{ldispl}%
where the method operations are defined as follows:
\begin{ldispl}
\begin{aeqns}
\Expii(x) & = & \tup{\True,2^x}\;, \\
\Factv(x) & = &
\tup{\True,\max \set{y \where \Exists{z}{x = 5^y \mmul z}}}
\end{aeqns}
\end{ldispl}%
and, for each $i \in [0,5]$:%
\footnote
{As usual, we write $x \divs y$ for $y$ is divisible by $x$.}
\begin{ldispl}
\begin{aeqns}
\rmmo{i}{succ}(x) & = & \tup{\True,\prim_i \mmul x}\;, \\
\rmmo{i}{pred}(x) & = &
\Biggl\{
\begin{array}[c]{@{}l@{\;\;}l@{}}
\tup{\True, x \mdiv \prim_i} & \mif \prim_i \divs x \\
\tup{\False,x}               & \mif \Not (\prim_i \divs x)\;,
\end{array}
\beqnsep
\rmmo{i}{iszero}(x) & = &
\Biggl\{
\begin{array}[c]{@{}l@{\;\;}l@{}}
\tup{\True, x} \phantom{{} \mdiv \prim_i}
                             & \mif \Not (\prim_i \divs x) \\
\tup{\False,x}               & \mif \prim_i \divs x\;,
\end{array}
\end{aeqns}
\end{ldispl}%
where $\prim_i$ is the $(i{+}1)$th prime number, i.e.\
$\prim_0 = 2$, $\prim_1 = 3$, $\prim_2 = 5$, \ldots\ .

We define a function $\rmlful$ from $\RML$ to $\Lf{\IF(\Univ)}$, which
gives, for each instruction sequence $P$ in $\RML$, the instruction
sequence in $\Lf{\IF(\Univ)}$ by which the effect produced by $P$
on a register machine with six registers can be simulated on $\Univ$.
This function is defined as follows:
\begin{ldispl}
\rmlful(u_1 \conc \ldots \conc u_k)
\\ \quad {} =
f.\expii \conc \phi(u_1) \conc \ldots \conc \phi(u_k) \conc
\ntst{f.\rmmn{1}{iszero}} \conc \fjmp{3} \conc
f.\factv \conc \haltP \conc f.\factv \conc \haltN\;,
\end{ldispl}%
where
\begin{ldispl}
\begin{aceqns}
\phi(a) & = & \psi(a)\;, \\
\phi(\ptst{a}) & = & \ptst{\psi(a)}\;, \\
\phi(\ntst{a}) & = & \ntst{\psi(a)}\;, \\
\phi(u) & = & u
 & \mif u \;\mathrm{is\;a\;jump\;or\;termination\;instruction}\;,
\end{aceqns}
\end{ldispl}%
where, for each $i \in [0,5]$:
\begin{ldispl}
\begin{aceqns}
\psi(\rmreg{i}.\rmincr)   & = & f.\rmmn{i}{succ}\;, \\
\psi(\rmreg{i}.\rmdecr)   & = & f.\rmmn{i}{pred}\;, \\
\psi(\rmreg{i}.\rmiszero) & = & f.\rmmn{i}{iszero}\;.
\end{aceqns}
\end{ldispl}%

Take an arbitrary computable $M \in \MO(\Nat)$.
Then there exists an instruction sequence in $\RML$ that computes $M$.
Take an arbitrary $P \in \RML$ that computes $M$.
Then $\moextr{\rmlful(P)}{\Univ} = M$.
Hence, $M$ is a derived method operation of $\Univ$.
\qed
\end{proof}
The universal computable functional unit $\Univ$ defined in the proof of
Theorem~\ref{theorem-universal-fu} has $20$ method operations.
However, three method operations suffice.
\begin{theorem}
\label{theorem-universal-fu-three-meths}
There exists a computable $\cH \in \FU(\Nat)$ with only three method
operations that is universal.
\end{theorem}
\begin{proof}
We know from the proof of Theorem~\ref{theorem-universal-fu} that there
exists a computable $\cH \in \FU(\Nat)$ with $20$ method operations, say
$M_0$, \ldots, $M_{19}$.
We will show that there exists a computable $\cH' \in \FU(\Nat)$ with
only three method operations such that $\cH \leq \cH'$.

We define a computable functional unit $\Univiii \in \FU(\Nat)$ with
only three method operations such that $\Univ \leq \Univiii$ as follows:
\begin{ldispl}
\Univiii = \set{\tup{\gi,\Gi},\tup{\gii,\Gii},\tup{\giii,\Giii}}\;,
\end{ldispl}%
where the method operations are defined as follows:
\begin{ldispl}
\begin{aeqns}
\Gi(x)   & = & \tup{\True,2^x}\;, \\
\Gii(x)  & = &
\left\{
\begin{array}[c]{@{}l@{\,}l@{}}
\tup{\True, 3 \mmul x}
 & \mif \Not (3^{19} \divs x) \And
        \Exists{y,z}{x = 3^y \mmul 2^z} \\
\tup{\True, x \mdiv 3^{19}}
 & \mif 3^{19} \divs x \And \Not (3^{20} \divs x) \And
        \Exists{y,z}{x = 3^y \mmul 2^z} \\
\tup{\False,0}
 & \mif 3^{20} \divs x \Or
        \Not \Exists{y,z}{x = 3^y \mmul 2^z}\;,
\end{array}
\right.
\beqnsep
\Giii(x) & = & M_{\factiii(x)}(\factii(x))\;,
\end{aeqns}
\end{ldispl}%
where
\begin{ldispl}
\begin{aeqns}
\factii(x)  & = & \max \set{y \where \Exists{z}{x = 2^y \mmul z}}\;, \\
\factiii(x) & = & \max \set{y \where \Exists{z}{x = 3^y \mmul z}}\;.
\end{aeqns}
\end{ldispl}%

We have that $M_i(x) = \Giii(3^i \mmul 2^x)$ for each $i \in [0,19]$.
Moreover, state $3^i \mmul 2^x$ can be obtained from state $x$ by first
applying $\Gi$ once and next applying $\Gii$ $i$ times.
Hence, for each $i \in [0,19]$,
$\moextr
  {f.\gi \conc {f.\gii}^{\,i} \conc
   \ptst{f.\giii} \conc \haltP \conc \haltN}
  {\Univiii} = M_i$.%
\footnote
{For each primitive instruction $u$, the instruction sequence $u^n$ is
 defined by induction on $n$ as follows: $u^0 = \fjmp{1}$, $u^1 = u$ and
 $u^{n+2} = u \conc u^{n+1}$.}
This means that $M_0$, \ldots, $M_{19}$ are derived method operations of
$\Univiii$.
\qed
\end{proof}
The universal computable functional unit $\Univiii$ defined in the proof
of Theorem~\ref{theorem-universal-fu-three-meths} has three method
operations.
We can show that one method operation does not suffice.
\begin{theorem}
\label{theorem-universal-fu-one-meth}
There does not exist a computable $\cH \in \FU(\Nat)$ with only one
method operation that is universal.
\end{theorem}
\begin{proof}
We will show that there does not exist a computable $\cH \in \FU(\Nat)$
with one method operation such that $\Counter \leq \cH$.
Here, $\Counter$ is the functional unit introduced at the beginning of
this section.

Assume that there exists a computable $\cH \in \FU(\Nat)$ with one
method operation such that $\Counter \leq \cH$.
Let $\cH' \in \FU(\Nat)$ be such that $\cH'$ has one method operation
and $\Counter \leq \cH'$, and let $m$ be the unique method name such
that $\IF(\cH') = \set{m}$.
Take arbitrary $P_1,P_2 \in \Lf{\IF(\cH')}$ such that
$\moextr{P_1}{\cH'} = \Incr$ and $\moextr{P_2}{\cH'} = \Decr$.
Then $\moextr{P_1}{\cH'}(0) = \tup{\True,1}$ and
$\moextr{P_2}{\cH'}(1) = \tup{\True,0}$.
Instruction $f.m$ is processed at least once if $P_1$ is applied to
$\cH'(0)$ or $P_2$ is applied to $\cH'(1)$.
Let $k_0$ be the number of times that instruction $f.m$ is processed on
application of $P_1$ to $\cH'(0)$ and
let $k_1$ be the number of times that instruction $f.m$ is processed on
application of $P_2$ to $\cH'(1)$ (irrespective of replies).
Then, from state $0$, state $0$ is reached again after $f.m$ is
processed $k_0 + k_1$ times.
Thus, by repeated application of $P_1$ to $\cH'(0)$ at most $k_0 + k_1$
different states can be reached.
This contradicts with $\moextr{P_1}{\cH'} = \Incr$.
Hence, there does not exist a computable $\cH \in \FU(\Nat)$ with one
method operation such that $\Counter \leq \cH$.
\qed
\end{proof}
It is an open problem whether two method operations suffice.

To the best of our knowledge, there are no existing results in
computability theory directly related to
Theorems~\ref{theorem-universal-fu},
\ref{theorem-universal-fu-three-meths}
and~\ref{theorem-universal-fu-one-meth}.
We could not even say which existing notion from computability theory
corresponds to the universality of a functional unit for natural
numbers.

In Section~\ref{sect-func-unit-sbs}, we will give a rough sketch of a
universal functional unit for a state space whose elements can be
understood as the possible contents of the tape of Turing machines
with a particular tape alphabet.
This universal functional unit corresponds to the common part of all
Turing machines with that tape alphabet.
The part that differs for different Turing machines is what is usually
called their ``transition function'' or ``program''.
In the current setting, the role of that part is filled by an
instruction sequence whose instructions correspond to the method
operations of the above-mentioned universal functional unit.
This means that different instruction sequences are needed for
different Turing machines with the tape alphabet concerned, but the
same universal functional unit suffices for all of them.
In particular, the same universal functional unit suffices for
universal Turing machines and non-universal Turing machines.

\section{Functional Units Relating to Turing Machine Tapes}
\label{sect-func-unit-sbs}

In this section, we define some notions that have a bearing on the
halting problem in the setting of \PGLBsbt\ and functional units.
The notions in question are defined in terms of functional units for the
following state space:
\begin{ldispl}
\SBS = \set{v \pebble w \where v,w \in \seqof{\set{0,1,\sep}}}\;.
\end{ldispl}%

The elements of $\SBS$ can be understood as the possible contents of the
tape of a Turing machine whose tape alphabet is $\set{0,1,\sep}$,
including the position of the tape head.
Consider an element $v \pebble w \in \SBS$.
Then $v$ corresponds to the content of the tape to the left of the
position of the tape head and $w$ corresponds to the content of the tape
from the position of the tape head to the right -- the indefinite
numbers of padding blanks at both ends are left out.
The colon serves as a separator of bit sequences.
This is for instance useful if the input of a program consists of
another program and an input to the latter program, both encoded as a
bit sequences.
We could have taken any other tape alphabet whose cardinality is greater
than one, but $\set{0,1,\sep}$ is extremely handy when dealing with
issues relating to the halting problem.
In fact, we could first have introduced the general notation $\SBS_A$,
where $A$ stands for a finite set of tape symbols, for the set
$\set{v \pebble w \where v,w \in \seqof{A}}$ and then have introduced
$\SBS$ as an abbreviation for $\SBS_{\set{0,1,\sep}}$.

Below, we will use a computable injective function
$\funct{\alpha}{\SBS}{\Nat}$ to encode the members of $\SBS$ as natural
numbers.
Because $\SBS$ is a countably infinite set, we assume that it is
understood what is a computable function from $\SBS$ to $\Nat$.
An obvious instance of a computable injective function
$\funct{\alpha}{\SBS}{\Nat}$ is the one where
$\alpha(a_1 \ldots a_n)$ is the natural number represented in the
quinary number-system by $a_1 \ldots a_n$ if the symbols $0$, $1$,
$\sep$ and $\pebble$ are taken as digits representing the numbers $1$,
$2$, $3$ and $4$, respectively.

A method operation $M \in \MO(\SBS)$ is \emph{computable} if there exist
computable functions $\funct{F,G}{\Nat}{\Nat}$ such that
$M(v) = \tup{\beta(F(\alpha(v))),\alpha^{-1}(G(\alpha(v)))}$ for all
$v \in \SBS$, where $\funct{\alpha}{\SBS}{\Nat}$ is a computable
injection and $\funct{\beta}{\Nat}{\Bool}$ is inductively defined by
$\beta(0) = \True$ and $\beta(n + 1) = \False$.
A functional unit $\cH \in \FU(\SBS)$ is \emph{computable} if, for each
$\tup{m,M} \in \cH$, $M$ is computable.

Like in the case of $\FU(\Nat)$, a computable $\cH \in \FU(\SBS)$ is
\emph{universal} if for each computable $\cH' \in \FU(\SBS)$, we have
$\cH' \below \cH$.

An example of a computable functional unit in $\FU(\SBS)$ is the
functional unit whose method operations correspond to the basic steps
that a Turing machine with tape alphabet $\set{0,1,\sep}$ can perform
on its tape.
It turns out that this functional unit is universal, which can be
proved using simple programming in \PGLBbt.

It is assumed that, for each $\cH \in \FU(\SBS)$, a computable
injective function from $\Lf{\IF(\cH)}$ to $\seqof{\set{0,1}}$ with a
computable image has been given that yields, for each
$x \in \Lf{\IF(\cH)}$, an encoding of $x$ as a bit sequence.
If we consider the case where the jump lengths in jump instructions are
character strings representing the jump lengths in decimal notation and
method names are character strings, such an encoding function can
easily be obtained using the ASCII character-encoding.
We use the notation $\ol{x}$ to denote the encoding of $x$ as a bit
sequence.

Let $\cH \in \FU(\SBS)$, and let $I \subseteq \IF(\cH)$.
Then:
\begin{itemize}
\item
$x \in \Lf{\IF(\cH)}$ produces a
\emph{solution of the halting problem} for $\Lf{I}$ with respect to
$\cH$ if:
\begin{ldispl}
x \cvg f.\cterm{\cH(v)}\; \mathrm{for\; all}\; v \in \SBS\;, \\
x \sfreply f.\cterm{\cH(\pebble \ol{y} \sep v)} = \True \Iff
y \cvg f.\cterm{\cH(\pebble v)}\; \mathrm{for\; all}\;
y \in \Lf{I}\; \mathrm{and}\; v \in \seqof{\set{0,1,\sep}}\;;
\end{ldispl}%
\item
$x \in \Lf{\IF(\cH)}$ produces a
\emph{reflexive solution of the halting problem} for $\Lf{I}$ with
respect to $\cH$ if $x$ produces a solution of the halting problem for
$\Lf{I}$ with respect to $\cH$ and $x \in \Lf{I}$;
\item
the halting problem for $\Lf{I}$ with respect to $\cH$ is
\emph{autosolvable} if there exists an $x \in \Lf{\IF(\cH)}$ such that
$x$ produces a reflexive solution of the halting problem for $\Lf{I}$
with respect to $\cH$;
\item
the halting problem for $\Lf{I}$ with respect to $\cH$ is
\emph{potentially autosolvable} if there exists an extension $\cH'$ of
$\cH$ such that the halting problem for $\Lf{\IF(\cH')}$ with respect
to $\cH'$ is autosolvable;
\item
the halting problem for $\Lf{I}$ with respect to $\cH$ is
\emph{potentially recursively autosolvable} if there exists an
extension $\cH'$ of $\cH$ such that the halting problem for
$\Lf{\IF(\cH')}$ with respect to $\cH'$ is autosolvable and $\cH'$ is
computable.
\end{itemize}
These definitions make clear that each combination of an
$\cH \in \FU(\SBS)$ and an $I \subseteq \IF(\cH)$ gives rise to a
\emph{halting problem instance}.

In Section~\ref{sect-interpreters} and~\ref{sect-autosolvability}, we
will make use of a method operation $\Dup \in \MO(\SBS)$ for duplicating
bit sequences.
This method operation is defined as follows:
\begin{ldispl}
\begin{aceqns}
\Dup(v \pebble w) & = & \Dup(\pebble v w)\;, \\
\Dup(\pebble v)   & = & \tup{\True,\pebble v \sep v}
 & \mif v \in \seqof{\set{0,1}}\;, \\
\Dup(\pebble v \sep w) & = & \tup{\True,\pebble v \sep v \sep w}
 & \mif v \in \seqof{\set{0,1}}\;.
\end{aceqns}
\end{ldispl}%

\begin{proposition}
\label{prop-dup}
Let $\cH \in \FU(\SBS)$ be such that $\tup{\dup,\Dup} \in \cH$,
let $I \subseteq  \IF(\cH)$ be such that $\dup \in I$,
let $x \in \Lf{I}$, and
let $v \in \seqof{\set{0,1}}$ and $w \in \seqof{\set{0,1,\sep}}$ be such
that $w = v$ or $w = v \sep w'$ for some $w' \in \seqof{\set{0,1,\sep}}$.
Then
$(f.\dup \conc x) \sfreply f.\cterm{\cH(\pebble w)} =
 x \sfreply f.\cterm{\cH(\pebble v \sep w)}$.
\end{proposition}
\begin{proof}
This follows immediately from the definition of $\Dup$ and the axioms
for~$\sfreply$.
\qed
\end{proof}
The method operation $\Dup$ is a derived method operation of the
above-mentioned functional unit whose method operations correspond to
the basic steps that a Turing machine with tape alphabet
$\set{0,1,\sep}$ can perform on its tape.
This follows immediately from the computability of $\Dup$ and the
universality of this functional unit.

In Sections~\ref{sect-interpreters} and~\ref{sect-autosolvability}, we
will make use of two simple transformations of \PGLBsbt\ instruction
sequences that affect only their termination behaviour on execution and
the Boolean value yielded at termination in the case of termination.
Here, we introduce notations for those transformations.

Let $x$ be a \PGLBsbt\ instruction sequence.
Then we write $\swap(x)$ for $x$ with each occurrence of $\haltP$
replaced by $\haltN$ and each occurrence of $\haltN$ replaced by
$\haltP$, and we write $\ftod(x)$ for $x$ with each occurrence of
$\haltN$ replaced by $\fjmp{0}$.
In the following proposition, the most important properties relating to
these transformations are stated.
\begin{proposition}
\label{prop-swap-f2d}
Let $x$ be a \PGLBsbt\ instruction sequence.
Then:
\begin{enumerate}
\item
if $x \sfreply u = \True$ then $\swap(x) \sfreply u = \False$ and
$\ftod(x) \sfreply u = \True$;
\item
if $x \sfreply u = \False$ then $\swap(x) \sfreply u = \True$ and
$\ftod(x) \sfreply u = \Div$.
\end{enumerate}
\end{proposition}
\begin{proof}
Let $p$ be a closed \BTAbt\ term of sort $\Thr$.
Then we write $\swap'(p)$ for $p$ with each occurrence of $\StopP$
replaced by $\StopN$ and each occurrence of $\StopN$ replaced by
$\StopP$, and we write $\ftod'(p)$ for $p$ with each occurrence of
$\StopN$ replaced by $\DeadEnd$.
It is easy to prove by induction on $i$ that
$\extr{i,\swap(x)} = \swap'(\extr{i,x})$ and
$\extr{i,\ftod(x)} = \ftod'(\extr{i,x})$ for all $i \in \Nat$.
By this result, Lemma~\ref{lemma-projections-BTAbt}, and axiom R10, it
is sufficient to prove the following for each closed \BTAbt\ term $p$ of
sort $\Thr$:
\begin{enumerate}
\item[]
if $p \sfreply u = \True$ then $\swap'(p) \sfreply u = \False$ and
$\ftod'(p) \sfreply u = \True$;
\item[]
if $p \sfreply u = \False$ then $\swap'(p) \sfreply u = \True$ and
$\ftod'(p) \sfreply u = \Div$.
\end{enumerate}
This is easy by induction on the structure of $p$.
\qed
\end{proof}

By the use of foci and the introduction of apply and reply operators on
service families, we make it possible to deal with cases that remind of
multi-tape Turing machines, Turing machines that has random access
memory, etc.
However, in this paper, we will only consider the case that reminds of
single-tape Turing machines.
This means that we will use only one focus ($f$) and only
singleton service families.

\section{Interpreters}
\label{sect-interpreters}

It is often mentioned in textbooks on computability that an
interpreter, which is a program for simulating the execution of
programs that it is given as input, cannot solve the halting problem
because the execution of the interpreter will not terminate if the
execution of its input program does not terminate.
In this section, we have a look at the termination behaviour of
interpreters in the setting of \PGLBsbt\ and functional units.

Let $\cH \in \FU(\SBS)$, let $I \subseteq \IF(\cH)$, and
let $I' \subseteq I$.
Then $x \in \Lf{I}$ is an \emph{interpreter} for $\Lf{I'}$ with respect
to $\cH$ if for all $y \in \Lf{I'}$ and $v \in \seqof{\set{0,1,\sep}}$:
\begin{ldispl}
y \cvg f.\cterm{\cH(\pebble v)} \Implies {} \\ \quad
x \cvg f.\cterm{\cH(\pebble \ol{y} \sep v)} \And
x \sfapply f.\cterm{\cH(\pebble \ol{y} \sep v)} =
y \sfapply f.\cterm{\cH(\pebble v)} \And
x \sfreply f.\cterm{\cH(\pebble \ol{y} \sep v)} =
y \sfreply f.\cterm{\cH(\pebble v)}\;.
\end{ldispl}%
Moreover, $x \in \Lf{I}$ is a \emph{reflexive interpreter} for $\Lf{I'}$
with respect to $\cH$ if $x$ is an interpreter for $\Lf{I'}$ with respect
to $\cH$ and $x \in \Lf{I'}$.

The following theorem states that a reflexive interpreter that always
terminates is impossible in the presence of the method operation $\Dup$.
\begin{theorem}
\label{theorem-interpreter}
Let $\cH \in \FU(\SBS)$ be such that $\tup{\dup,\Dup} \in \cH$,
let $I \subseteq \IF(\cH)$ be such that $\dup \in I$, and
let $x \in \Lf{\IF(\cH)}$ be a reflexive interpreter for $\Lf{I}$ with
respect to $\cH$.
Then there exist an $y \in \Lf{I}$ and a $v \in \seqof{\set{0,1,\sep}}$
such that $x \dvg f.\cterm{\cH(\pebble \ol{y} \sep v)}$.
\end{theorem}
\begin{proof}
Assume the contrary.
Take $y = f.\dup \conc \swap(x)$.
By the assumption, $x \cvg f.\cterm{\cH(\pebble \ol{y} \sep \ol{y})}$.
By Propositions~\ref{prop-cvg-sfreply} and~\ref{prop-swap-f2d}, it
follows that $\swap(x) \cvg f.\cterm{\cH(\pebble \ol{y} \sep \ol{y})}$
and
$\swap(x) \sfreply f.\cterm{\cH(\pebble \ol{y} \sep \ol{y})} \neq
 x \sfreply f.\cterm{\cH(\pebble \ol{y} \sep \ol{y})}$.
By Propositions~\ref{prop-cvg-sfreply} and~\ref{prop-dup}, it follows
that $(f.\dup \conc \swap(x)) \cvg f.\cterm{\cH(\pebble \ol{y})}$ and
$(f.\dup \conc \swap(x)) \sfreply f.\cterm{\cH(\pebble \ol{y})} \neq
 x \sfreply f.\cterm{\cH(\pebble \ol{y} \sep \ol{y})}$.
Since $y = f.\dup \conc \swap(x)$, we have
$y \cvg f.\cterm{\cH(\pebble \ol{y})}$ and
$y \sfreply f.\cterm{\cH(\pebble \ol{y})} \neq
 x \sfreply f.\cterm{\cH(\pebble \ol{y} \sep \ol{y})}$.
Because $x$ is a reflexive interpreter, this implies
$x \sfreply f.\cterm{\cH(\pebble \ol{y} \sep \ol{y})} =
 y \sfreply f.\cterm{\cH(\pebble \ol{y})}$ and
$y \sfreply f.\cterm{\cH(\pebble \ol{y})} \neq
 x \sfreply f.\cterm{\cH(\pebble \ol{y} \sep \ol{y})}$.
This is a contradiction.
\qed
\end{proof}
It is easy to see that Theorem~\ref{theorem-interpreter} goes through
for all functional units for $\SBS$ of which $\Dup$ is a derived method
operation.
Recall that the functional units concerned include the afore-mentioned
functional unit whose method operations correspond to the basic steps
that a Turing machine with tape alphabet $\set{0,1,\sep}$ can perform
on its tape.

For each $\cH \in \FU(\SBS)$, $m \in \IF(\cH)$, and $v \in \SBS$,
we have $(f.m \conc \haltP \conc \haltN) \cvg f.\cterm{\cH(v)}$.
This leads us to the following corollary of
Theorem~\ref{theorem-interpreter}.
\begin{corollary}
\label{corollary-interpreter}
For all $\cH \in \FU(\SBS)$ with $\tup{\dup,\Dup} \in \cH$ and
$I \subseteq \IF(\cH)$ with $\dup \in I$, there does not exist an
$m \in I$ such that $f.m \conc \haltP \conc \haltN$ is a reflexive
interpreter for $\Lf{I}$ with respect to $\cH$.
\end{corollary}

To the best of our knowledge, there are no existing results in
computability theory or elsewhere directly related to
Theorem~\ref{theorem-interpreter}.
It looks as if the closest to this result are results on termination of
particular interpreters for particular logic and functional programming
languages.

\section{Autosolvability of the Halting Problem}
\label{sect-autosolvability}

Because a reflexive interpreter that always terminates is impossible in
the presence of the method operation $\Dup$, we must conclude that
solving the halting problem by means of a reflexive interpreter is out
of the question in the presence of the method operation $\Dup$.
The question arises whether the proviso ``by means of a reflexive
interpreter'' can be dropped.
In this section, we answer this question in the affirmative.
Before we present this negative result concerning autosolvability of the
halting problem, we present a positive result.

Let $M \in \MO(\SBS)$.
Then we say that $M$ \emph{increases the number of colons} if for some
$v \in \SBS$ the number of colons in $M^e(v)$ is greater than the number
of colons in $v$.

\begin{theorem}
\label{theorem-autosolv}
Let $\cH \in \FU(\SBS)$ be such that no method operation of $\cH$
increases the number of colons.
Then there exist an extension $\cH'$ of $\cH$,
an $I' \subseteq \IF(\cH')$, and an $x \in \Lf{\IF(\cH')}$ such that
$x$ produces a reflexive solution of the halting problem for $\Lf{I'}$
with respect to $\cH'$.
\end{theorem}
\begin{proof}
Let $\halting \in \MN$ be such that $\halting \notin \IF(\cH)$.
Take $I' = \IF(\cH) \union \set{\halting}$.
Take $\cH' = \cH \union \set{\tup{\halting,\Halting}}$, where
$\Halting \in \MO(\SBS)$ is defined as follows:
\begin{ldispl}
\begin{aceqns}
\Halting(v \pebble w) & = & \Halting(\pebble v w)\;, \\
\Halting(\pebble v)   & = & \tup{\False,\pebble}
 & \mif v \in \seqof{\set{0,1}}\;, \\
\Halting(\pebble v \sep w) & = & \tup{\False,\pebble}
 & \mif v \in \seqof{\set{0,1}} \And
        \Forall{x \in \Lf{I'}}{v \neq \ol{x}}\;, \\
\Halting(\pebble \ol{x} \sep w) & = & \tup{\False,\pebble}
 & \mif x \in \Lf{I'} \And x \dvg f.\cterm{\cH'(w)}\;, \\
\Halting(\pebble \ol{x} \sep w) & = & \tup{\True,\pebble}
 & \mif x \in \Lf{I'} \And x \cvg f.\cterm{\cH'(w)}\;.
\end{aceqns}
\end{ldispl}%
Then $\ptst{f.\halting} \conc \haltP \conc \haltN$ produces a reflexive
solution of the halting problem for $\Lf{I'}$ with respect to $\cH'$.
\qed
\end{proof}
Theorem~\ref{theorem-autosolv} tells us that there exist functional
units $\cH \in \FU(\SBS)$ with the property that the halting problem is
potentially autosolvable for $\Lf{\IF(\cH)}$ with respect to $\cH$.
Thus, we know that there exist functional units $\cH \in \FU(\SBS)$ with
the property that the halting problem is autosolvable for
$\Lf{\IF(\cH)}$ with respect to~$\cH$.

There exists an $\cH \in \FU(\SBS)$ for which $\Halting$ as defined in
the proof of Theorem~\ref{theorem-autosolv} is computable.
\begin{theorem}
\label{theorem-comput}
Let $\cH = \emptyset$ and
$\cH' = \cH \union \set{\tup{\halting,\Halting}}$, where
$\Halting$ is as defined in the proof
of Theorem~\ref{theorem-autosolv}.
Then, $\Halting$ is computable.
\end{theorem}
\begin{proof}
It is sufficient to prove for an arbitrary $x \in \Lf{\IF(\cH')}$ that,
for all $v \in \SBS$, $x \cvg f.\cterm{\cH'(v)}$ is decidable.
We will prove this by induction on the number of colons in $v$.

The basis step.
Because the number of colons in $v$ equals $0$,
$\Halting(v) = \tup{\False,\pebble}$.
It follows that $x \cvg f.\cterm{\cH'(v)} \Iff x' \cvg \emptysf$,
where $x'$ is $x$ with each occurrence of $f.\halting$ and
$\ntst{f.\halting}$ replaced by $\fjmp{1}$ and each occurrence of
$\ptst{f.\halting}$ replaced by $\fjmp{2}$.
Because $x'$ is finite, $x' \cvg \emptysf$ is decidable.
Hence, $x \cvg f.\cterm{\cH'(v)}$ is decidable.

The inductive step.
Because the number of colons in $v$ is greater than $0$, either
$\Halting(v) = \tup{\True,\pebble}$ or
$\Halting(v) = \tup{\False,\pebble}$.
It follows that $x \cvg f.\cterm{\cH'(v)} \Iff x' \cvg \emptysf$, where
$x'$ is $x$ with:
\begin{itemize}
\item
each occurrence of $f.\halting$ and $\ptst{f.\halting}$ replaced by
$\fjmp{1}$ if the occurrence leads to the first application of
$\Halting$ and $\Halting^r(v) = \True$, and by $\fjmp{2}$ otherwise;
\item
each occurrence of $\ntst{f.\halting}$ replaced by
$\fjmp{2}$ if the occurrence leads to the first application of
$\Halting$ and $\Halting^r(v) = \True$, and by $\fjmp{1}$ otherwise.
\end{itemize}
An occurrence of $f.\halting$, $\ptst{f.\halting}$ or
$\ntst{f.\halting}$ in $x$ leads to the first application of $\Halting$
iff $\extr{1,x} = \extr{i,x}$, where $i$ is the position of the
occurrence in $x$.
Because $x$ is finite, it is decidable whether an occurrence of
$f.\halting$, $\ptst{f.\halting}$ or $\ntst{f.\halting}$ leads to the
first processing of $\halting$.
Moreover, by the induction hypothesis, it is decidable whether
$\Halting^r(v) = \True$.
Because $x'$ is finite, it follows that $x' \cvg \emptysf$ is decidable.
Hence, $x \cvg f.\cterm{\cH'(v)}$ is decidable.
\qed
\end{proof}
Theorems~\ref{theorem-autosolv} and~\ref{theorem-comput} together tell
us that there exists a functional unit $\cH \in \FU(\SBS)$, viz.\
$\emptyset$, with the property that the halting problem is potentially
recursively autosolvable for $\Lf{\IF(\cH)}$ with respect to $\cH$.

Let $\cH \in \FU(\SBS)$ be such that all derived method operations of
$\cH$ are computable and do not increase the number of colons.
Then the halting problem is potentially autosolvable for $\Lf{\IF(\cH)}$
with respect to $\cH$.
However, the halting problem is not always potentially recursively
autosolvable for $\Lf{\IF(\cH)}$ with respect to $\cH$ because otherwise
the halting problem would always be decidable.

The following theorem tells us essentially that potential
autosolvability of the halting problem is precluded in the presence of
the method operation $\Dup$.
\begin{theorem}
\label{theorem-non-autosolv}
Let $\cH \in \FU(\SBS)$ be such that $\tup{\dup,\Dup} \in \cH$, and
let $I \subseteq \IF(\cH)$ be such that $\dup \in I$.
Then there does not exist an $x \in \Lf{\IF(\cH)}$ such that $x$
produces a reflexive solution of the halting problem for $\Lf{I}$ with
respect to $\cH$.
\end{theorem}
\begin{proof}
Assume the contrary.
Let $x \in \Lf{\IF(\cH)}$ be such that $x$ produces a reflexive
solution of the halting problem for $\Lf{I}$ with respect to $\cH$, and
let $y = f.\dup \conc \ftod(\swap(x))$.
Then $x \cvg f.\cterm{\cH(\pebble \ol{y} \sep \ol{y})}$.
By Propositions~\ref{prop-cvg-sfreply} and~\ref{prop-swap-f2d}, it
follows that $\swap(x) \cvg f.\cterm{\cH(\pebble \ol{y} \sep \ol{y})}$
and either
$\swap(x) \sfreply f.\cterm{\cH(\pebble \ol{y} \sep \ol{y})} = \True$ or
$\swap(x) \sfreply f.\cterm{\cH(\pebble \ol{y} \sep \ol{y})} = \False$.

In the case where
$\swap(x) \sfreply f.\cterm{\cH(\pebble \ol{y} \sep \ol{y})} = \True$,
we have by Proposition~\ref{prop-swap-f2d} that
(i)~$\ftod(\swap(x)) \sfreply f.\cterm{\cH(\pebble \ol{y} \sep \ol{y})}
       = \True$ and
(ii)~$x \sfreply f.\cterm{\cH(\pebble \ol{y} \sep \ol{y})} = \False$.
By Proposition~\ref{prop-dup}, it follows from~(i) that
$(f.\dup \conc \ftod(\swap(x))) \sfreply f.\cterm{\cH(\pebble \ol{y})} =
 \True$.
Since $y = f.\dup \conc \ftod(\swap(x))$, we have
$y \sfreply f.\cterm{\cH(\pebble \ol{y})} = \True$.
On the other hand, because $x$ produces a reflexive solution, it follows
from~(ii) that $y \dvg f.\cterm{\cH(\pebble \ol{y})}$.
By Proposition~\ref{prop-cvg-sfreply}, this contradicts with
$y \sfreply f.\cterm{\cH(\pebble \ol{y})} = \True$.

In the case where
$\swap(x) \sfreply f.\cterm{\cH(\pebble \ol{y} \sep \ol{y})} = \False$,
we have by Proposition~\ref{prop-swap-f2d} that
(i)~$\ftod(\swap(x)) \sfreply f.\cterm{\cH(\pebble \ol{y} \sep \ol{y})}
       = \Div$ and
(ii)~$x \sfreply f.\cterm{\cH(\pebble \ol{y} \sep \ol{y})} =
\True$.
By Proposition~\ref{prop-dup}, it follows from~(i) that
$(f.\dup \conc \ftod(\swap(x))) \sfreply f.\cterm{\cH(\pebble \ol{y})} =
 \Div$.
Since $y = f.\dup \conc \ftod(\swap(x))$, we have
$y \sfreply f.\cterm{\cH(\pebble \ol{y})} = \Div$.
On the other hand, because $x$ produces a reflexive solution, it follows
from~(ii) that $y \cvg f.\cterm{\cH(\pebble \ol{y})}$.
By Proposition~\ref{prop-cvg-sfreply}, this contradicts with
$y \sfreply f.\cterm{\cH(\pebble \ol{y})} = \Div$.
\qed
\end{proof}
It is easy to see that Theorem~\ref{theorem-non-autosolv} goes through
for all functional units for $\SBS$ of which $\Dup$ is a derived method
operation.
Recall that the functional units concerned include the afore-mentioned
functional unit whose method operations correspond to the basic steps
that a Turing machine with tape alphabet $\set{0,1,\sep}$ can perform
on its tape.
Because of this, the unsolvability of the halting problem for Turing
machines can be understood as a corollary of
Theorem~\ref{theorem-non-autosolv}.

Below, we will give an alternative proof of
Theorem~\ref{theorem-non-autosolv}.
A case distinction is needed in both proofs, but in the alternative
proof it concerns a minor issue.
The issue in question is covered by the following lemma.
\begin{lemma}
\label{lemma-non-autosolv}
Let $\cH \in \FU(\SBS)$, let $I \subseteq \IF(\cH)$,
let $x \in \Lf{\IF(\cH)}$ be such that $x$ produces a reflexive solution
of the halting problem for $\Lf{I}$ with respect to $\cH$,
let $y \in \Lf{I}$, and let $v \in \seqof{\set{0,1,\sep}}$.
Then
$y \cvg f.\cterm{\cH(\pebble v)}$ implies
$y \sfreply f.\cterm{\cH(\pebble v)} =
 x \sfreply f.\cterm{\cH(\pebble \ol{\ftod(y)} \sep v)}$.
\end{lemma}
\begin{proof}
By Proposition~\ref{prop-cvg-sfreply}, it follows from
$y \cvg f.\cterm{\cH(\pebble v)}$ that either
$y \sfreply f.\cterm{\cH(\pebble v)} = \True$ or
$y \sfreply f.\cterm{\cH(\pebble v)} = \False$.

In the case where $y \sfreply f.\cterm{\cH(\pebble v)} = \True$, we have
by Propositions~\ref{prop-cvg-sfreply} and~\ref{prop-swap-f2d} that
$\ftod(y) \cvg f.\cterm{\cH(\pebble v)}$ and so
$x \sfreply f.\cterm{\cH(\pebble \ol{\ftod(y)} \sep v)} = \True$.

In the case where $y \sfreply f.\cterm{\cH(\pebble v)} = \False$, we have
by Propositions~\ref{prop-cvg-sfreply} and~\ref{prop-swap-f2d} that
$\ftod(y) \dvg f.\cterm{\cH(\pebble v)}$ and so
$x \sfreply f.\cterm{\cH(\pebble \ol{\ftod(y)} \sep v)} = \False$.
\qed
\end{proof}

\begin{proof}
[Another proof of Theorem~\ref{theorem-non-autosolv}.]
Assume the contrary.
Let $x \in \Lf{\IF(\cH)}$ be such that $x$ produces a reflexive solution
of the halting problem for $\Lf{I}$ with re\-spect to $\cH$, and
let $y = \ftod(\swap(f.\dup \conc x))$.
Then $x \cvg f.\cterm{\cH(\pebble \ol{y} \sep \ol{y})}$.
By Propo\-sitions~\ref{prop-cvg-sfreply}, \ref{prop-dup}
and~\ref{prop-swap-f2d}, it follows that
$\swap(f.\dup \conc x) \cvg f.\cterm{\cH(\pebble \ol{y})}$.
By Lemma~\ref{lemma-non-autosolv}, it follows that
$\swap(f.\dup \conc x) \sfreply f.\cterm{\cH(\pebble \ol{y})} =
  x \sfreply f.\cterm{\cH(\pebble \ol{y} \sep \ol{y})}$.
By Proposition~\ref{prop-swap-f2d}, it follows that
$(f.\dup \conc x) \sfreply f.\cterm{\cH(\pebble \ol{y})} \neq
 x \sfreply f.\cterm{\cH(\pebble \ol{y} \sep \ol{y})}$.
On the other hand, by Propo\-sition~\ref{prop-dup}, we have that
$(f.\dup \conc x) \sfreply f.\cterm{\cH(\pebble \ol{y})} =
 x \sfreply f.\cterm{\cH(\pebble \ol{y} \sep \ol{y})}$.
This contradicts with
$(f.\dup \conc x) \sfreply f.\cterm{\cH(\pebble \ol{y})} \neq
 x \sfreply f.\cterm{\cH(\pebble \ol{y} \sep \ol{y})}$.
\qed
\end{proof}
Both proofs of Theorem~\ref{theorem-non-autosolv} given above are
diagonalization proofs in disguise.

Now, let $\cH = \set{\tup{\dup,\Dup}}$.
By Theorem~\ref{theorem-non-autosolv}, the halting problem for
$\Lf{\set{\dup}}$ with respect to $\cH$ is not (potentially)
autosolvable.
However, it is decidable.
\begin{theorem}
\label{theorem-decidable}
Let $\cH = \set{\tup{\dup,\Dup}}$.
Then the halting problem for $\Lf{\set{\dup}}$ with respect to $\cH$ is
decidable.
\end{theorem}
\begin{proof}
Let $x \in \Lf{\set{\dup}}$, and let $x'$ be $x$ with each occurrence of
$f.\dup$ and $\ptst{f.\dup}$ replaced by $\fjmp{1}$ and each occurrence
of $\ntst{f.\dup}$ replaced by $\fjmp{2}$.
For all $v \in \SBS$, $\Dup^r(v) = \True$.
Therefore, $x \cvg f.\cH(v) \Iff x' \cvg \emptysf$ for all $v \in \SBS$.
Because $x'$ is finite, $x' \cvg \emptysf$ is decidable.
\qed
\end{proof}

It follows from Theorem~\ref{theorem-decidable} that there exists a
computable method operation by means of which a solution for the
halting problem for $\Lf{\set{\dup}}$ can be produced.
This leads us to the following corollary of
Theorem~\ref{theorem-decidable}.
\begin{corollary}
\label{corollary-computable}
There exist a computable $\cH \in \FU(\SBS)$ with
$\tup{\dup,\Dup} \in \cH$, an $I \subseteq \IF(\cH)$ with $\dup \in I$,
and an $x \in \Lf{\IF(\cH)}$ such that $x$ produces a solution of the
halting problem for $\Lf{I}$ with respect to $\cH$.
\end{corollary}

To the best of our knowledge, there are no existing results in
computability theory directly related to
Theorems~\ref{theorem-autosolv}, \ref{theorem-comput},
\ref{theorem-non-autosolv} and~\ref{theorem-decidable}.
The closest to these results are probably the positive results in the
setting of Turing machines that have been obtained with restrictions on
the number of states, the minimum of the number of transitions where
the tape head moves to the left and the number of transitions where the
tape head moves to the right, or the number of different combinations
of input symbol, direction of head move, and output symbol occurring
in the transitions (see e.g.~\cite{Pav73a,Mar97a}).

\section{Concluding Remarks}
\label{sect-concl}

We have presented a re-design of the extension of basic thread algebra
that was used in previous work to deal with the interaction between
instruction sequences under execution and components of their execution
environment concerning the processing of instructions.
The changes introduced allow for the material from quite a part of
that work to be streamlined.
Moreover, we have introduced the notion of a functional unit.
Using the resulting setting, we have obtained a novel computability
result about functional units for natural numbers and several novel
results relating to the autosolvability of the halting problem.

The following remarks may clarify the relationship between the setting
that is used in this paper and the setting of Turing machines and the
extent to which the results presented in this paper can be transferred
to the setting of Turing machines.

Each single-tape Turing machine can be simulated by means of a thread
that interacts with a service from a singleton service family.
The thread and service correspond to the finite control and tape of the
single-tape Turing machine.
The threads that correspond to the finite controls of single-tape
Turing machines are examples of regular threads, i.e.\ threads that can
only evolve into a finite number of other threads.
Similar remarks can be made about multi-tape Turing machines, register
machines, multi-stack machines, et cetera.

The results about functional units can probably be transferred to the
setting of Turing machines after the notion of a functional unit has
been linked with that setting.
However, we believe that the setting of Turing machines does not lend
itself well to the investigation of the universality of functional
units for natural numbers.
The results relating to the autosolvability of the halting problem
cannot be transferred to the setting of Turing machines because that
setting corresponds to a restriction to a single fixed functional unit
in our setting.
The point is that all Turing machines have the same tape manipulation
features.
Because of that only the effects of restrictions on the use of these
features on the solvability of the halting problem are open
for investigation in the setting of Turing machines.

The following remarks touch on closely related previous work on the
halting problem and an interesting option for related future work on
the halting problem.

The results relating to the autosolvability of the halting problem
extend and strengthen the results regarding the halting problem for
programs given in~\cite{BP04a} in a setting which looks to be more
adequate to describe and analyse issues regarding the halting problem
for programs.
It happens that decidability depends on the halting problem instance
considered.
This is different in the case of the on-line halting problem for
programs, i.e.\ the problem to forecast during its execution whether a
program will eventually terminate (see~\cite{BP04a}).

The bounded halting problem for programs is the problem to determine,
given a program and an input to the program, whether execution of the
program on that input will terminate after no more than a fixed number
of steps.
An interesting option for future work is to investigate whether we can
find a lower bound for the complexity of solving the bounded halting
problem for programs using an appropriate functional unit.

The following remarks are miscellaneous ones relating to the material
presented in the current paper.

We have proposed three instruction sequence processing operators: the
use operator, the apply operator and the reply operator.
The apply operator fits in with the viewpoint that programs are state
transformers that can be modelled by partial functions.
This viewpoint was first taken in the early days of denotational
semantics, see e.g.~\cite{Mos74a,Sto77b,Ten77a}.

Pursuant to~\cite{BM09g}, we have also proposed to comply with
conventions that exclude the use of terms that can be built by means of
the proposed operators, but are not really intended to denote anything.
The idea to comply with such conventions looks to be more widely
applicable in theoretical computer science.

In the case where the state space is $\Bool$, the state space consists
of only two states.
Because there are four possible unary functions on $\Bool$, there are
precisely $16$ method operations in $\MO(\Bool)$.
There are in principle $2^{16}$ different functional units in
$\FU(\Bool)$, for it is useless to include the same method operation
more than once under different names in a functional unit.
This means that $2^{16}$ is an upper bound of the number of functional
unit degrees in $\FU(\Bool) \mdiv {\equiv}$.
However, it is straightforward to show that $\FU(\Bool) \mdiv {\equiv}$
has only $12$ different functional unit degrees.
In the more general case of a finite state space consisting of $k$
states, say $S_k$, there are in principle $2^{2^k \mmul k^k}$ different
functional units in $\FU(S_k)$.
Already with $k = 3$, it becomes unclear whether the number of
functional unit degrees in $\FU(S_k)$ can be determined manually.
Actually, we do not know at the moment whether it can be determined with
computer support either.

\begin{acknowledgement}
We thank two anonymous referees for carefully reading preliminary
versions of this paper and for suggesting improvements of the
presentation of the paper.
\end{acknowledgement}

\bibliographystyle{spmpsci}
\bibliography{IS}

\begin{thebibliography}{10}
\providecommand{\url}[1]{{#1}}
\providecommand{\urlprefix}{URL }
\expandafter\ifx\csname urlstyle\endcsname\relax
  \providecommand{\doi}[1]{DOI~\discretionary{}{}{}#1}\else
  \providecommand{\doi}{DOI~\discretionary{}{}{}\begingroup
  \urlstyle{rm}\Url}\fi

\bibitem{AB09a}
Arora, S., Barak, B.: Computational Complexity: A Modern Approach.
\newblock Cambridge University Press, Cambridge (2009)

\bibitem{Bak91a}
Baker, H.G.: Precise instruction scheduling without a precise machine model.
\newblock SIGARCH Computer Architecture News \textbf{19}(6), 4--8 (1991)

\bibitem{BB03a}
Bergstra, J.A., Bethke, I.: Polarized process algebra and program equivalence.
\newblock In: J.C.M. Baeten, J.K. Lenstra, J.~Parrow, G.J. Woeginger (eds.)
  Proceedings 30th ICALP, \emph{Lecture Notes in Computer Science}, vol. 2719,
  pp. 1--21. Springer-Verlag (2003)

\bibitem{BL00a}
Bergstra, J.A., Loots, M.E.: Program algebra for component code.
\newblock Formal Aspects of Computing \textbf{12}(1), 1--17 (2000)

\bibitem{BL02a}
Bergstra, J.A., Loots, M.E.: Program algebra for sequential code.
\newblock Journal of Logic and Algebraic Programming \textbf{51}(2), 125--156
  (2002)

\bibitem{BM06a}
Bergstra, J.A., Middelburg, C.A.: A thread algebra with multi-level strategic
  interleaving.
\newblock Theory of Computing Systems \textbf{41}(1), 3--32 (2007)

\bibitem{BM08g}
Bergstra, J.A., Middelburg, C.A.: Instruction sequences and non-uniform
  complexity theory.
\newblock {\tt arXiv:0809.0352v3 [cs.CC]} (2008)

\bibitem{BM07g}
Bergstra, J.A., Middelburg, C.A.: Program algebra with a jump-shift
  instruction.
\newblock Journal of Applied Logic \textbf{6}(4), 553--563 (2008)

\bibitem{BM09m}
Bergstra, J.A., Middelburg, C.A.: Autosolvability of halting problem instances
  for instruction sequences.
\newblock {\tt arXiv:0911.5018v3 [cs.LO]} (2009)

\bibitem{BM09l}
Bergstra, J.A., Middelburg, C.A.: Functional units for natural numbers.
\newblock {\tt arXiv:0911.1851v3 [cs.PL]} (2009)

\bibitem{BM09i}
Bergstra, J.A., Middelburg, C.A.: Indirect jumps improve instruction sequence
  performance.
\newblock {\tt arXiv:0909.2089v2 [cs.PL]} (2009)

\bibitem{BM09k}
Bergstra, J.A., Middelburg, C.A.: Instruction sequence processing operators.
\newblock {\tt arXiv:0910.5564v4 [cs.LO]} (2009)

\bibitem{BM07c}
Bergstra, J.A., Middelburg, C.A.: On the operating unit size of load/store
  architectures.
\newblock Mathematical Structures in Computer Science \textbf{20}(3), 395--417
  (2010)

\bibitem{BM06c}
Bergstra, J.A., Middelburg, C.A.: A thread calculus with molecular dynamics.
\newblock Information and Computation \textbf{208}(7), 817--844 (2010)

\bibitem{BM09g}
Bergstra, J.A., Middelburg, C.A.: Inversive meadows and divisive meadows.
\newblock Journal of Applied Logic \textbf{9}(3), 203--220 (2011)

\bibitem{BM11c}
Bergstra, J.A., Middelburg, C.A.: On the behaviours produced by instruction
  sequences under execution.
\newblock {\tt arXiv:1106.6196v1 [cs.PL]} (2011)

\bibitem{BM08b}
Bergstra, J.A., Middelburg, C.A.: Thread extraction for polyadic instruction
  sequences.
\newblock Scientific Annals of Computer Science \textbf{21}(2), 283--310 (2011)

\bibitem{BM08h}
Bergstra, J.A., Middelburg, C.A.: On the expressiveness of single-pass
  instruction sequences.
\newblock Theory of Computing Systems \textbf{50}(2), 313--328 (2012)

\bibitem{BP02a}
Bergstra, J.A., Ponse, A.: Combining programs and state machines.
\newblock Journal of Logic and Algebraic Programming \textbf{51}(2), 175--192
  (2002)

\bibitem{BP04a}
Bergstra, J.A., Ponse, A.: Execution architectures for program algebra.
\newblock Journal of Applied Logic \textbf{5}(1), 170--192 (2007)

\bibitem{BP09a}
Bergstra, J.A., Ponse, A.: An instruction sequence semigroup with involutive
  anti-automorphisms.
\newblock Scientific Annals of Computer Science \textbf{19}, 57--92 (2009)

\bibitem{BT07a}
Bergstra, J.A., Tucker, J.V.: The rational numbers as an abstract data type.
\newblock Journal of the ACM \textbf{54}(2), Article 7 (2007)

\bibitem{BH97a}
Brock, C., Hunt, W.A.: Formally specifying and mechanically verifying programs
  for the {Motorola} complex arithmetic processor {DSP}.
\newblock In: ICCD '97, pp. 31--36 (1997)

\bibitem{HJP82a}
Hennessy, J., Jouppi, N., Przybylski, S., Rowen, C., Gross, T., Baskett, F.,
  Gill, J.: {MIPS}: A microprocessor architecture.
\newblock In: MICRO '82, pp. 17--22 (1982)

\bibitem{Her65a}
Hermes, H.: Enumerability, Decidability, Computability.
\newblock Springer-Verlag, Berlin (1965)

\bibitem{Kle36a}
Kleene, S.C.: General recursive functions of natural numbers.
\newblock Mathematische Annalen \textbf{112}, 727--742 (1936)

\bibitem{Lun77a}
Lunde, A.: Empirical evaluation of some features of instruction set processor
  architectures.
\newblock Communications of the ACM \textbf{20}(3), 143--153 (1977)

\bibitem{LB81a}
Lynch, N.A., Blum, E.K.: Relative complexity of algebras.
\newblock Mathematical Systems Theory \textbf{14}(1), 193--214 (1981)

\bibitem{Mar97a}
Margenstern, M.: Decidability and undecidability of the halting problem on
  {Turing} machines, a survey.
\newblock In: S.~Adian, A.~Nerode (eds.) LFCS'97, \emph{Lecture Notes in
  Computer Science}, vol. 1234, pp. 226--236. Springer-Verlag (1997)

\bibitem{Min61a}
Minsky, M.L.: Recursive unsolvability of {Post's} problem of ``tag'' and other
  topics in theory of {Turing} machines.
\newblock Annals of Mathematics \textbf{74}(3), 437--455 (1961)

\bibitem{Mos74a}
Mosses, P.D.: The mathematical semantics of {ALGOL 60}.
\newblock Tech. Rep. PRG-12, Programming Research Group, Oxford University
  (1974)

\bibitem{Mos06a}
Mosses, P.D.: Formal semantics of programming languages --- an overview.
\newblock Electronic Notes in Theoretical Computer Science \textbf{148}, 41--73
  (2006)

\bibitem{NH97a}
Nair, R., Hopkins, M.E.: Exploiting instruction level parallelism in processors
  by caching scheduled groups.
\newblock SIGARCH Computer Architecture News \textbf{25}(2), 13--25 (1997)

\bibitem{OH00a}
Ofelt, D., Hennessy, J.L.: Efficient performance prediction for modern
  microprocessors.
\newblock In: SIGMETRICS '00, pp. 229--239 (2000)

\bibitem{PD80a}
Patterson, D.A., Ditzel, D.R.: The case for the reduced instruction set
  computer.
\newblock SIGARCH Computer Architecture News \textbf{8}(6), 25--33 (1980)

\bibitem{Pav73a}
Pavlotskaya, L.M.: Solvability of the halting problem for certain classes of
  {Turing} machines.
\newblock Mathematical Notes \textbf{13}(6), 537--541 (1973)

\bibitem{PZ06a}
Ponse, A., van~der Zwaag, M.B.: An introduction to program and thread algebra.
\newblock In: A.~Beckmann, et~al. (eds.) CiE 2006, \emph{Lecture Notes in
  Computer Science}, vol. 3988, pp. 445--458. Springer-Verlag (2006)

\bibitem{ST99a}
Sannella, D., Tarlecki, A.: Algebraic preliminaries.
\newblock In: E.~Astesiano, H.J. Kreowski, B.~Krieg-Br{\"{u}}ckner (eds.)
  Algebraic Foundations of Systems Specification, pp. 13--30. Springer-Verlag,
  Berlin (1999)

\bibitem{SS63a}
Shepherdson, J.C., Sturgis, H.E.: Computability of recursive functions.
\newblock Journal of the ACM \textbf{10}(2), 217--255 (1963)

\bibitem{Sip06a}
Sipser, M.: Introduction to the Theory of Computation, second edn.
\newblock Thomson, Boston, MA (2006)

\bibitem{Sto77b}
Stoy, J.E.: Denotational Semantics: The {Scott-Strachey} Approach to
  Programming Language Theory.
\newblock Series in Computer Science. MIT Press, Cambridge, MA (1977)

\bibitem{TW07a}
Tennenhouse, D.L., Wetherall, D.J.: Towards an active network architecture.
\newblock SIGCOMM Computer Communication Review \textbf{37}(5), 81--94 (2007)

\bibitem{Ten77a}
Tennent, R.D.: A denotational definition of the programming language {Pascal}.
\newblock Tech. Rep. TR77-47, Department of Computing and Information Sciences,
  Queen's University, Kingston, Ontario, Canada (1977)

\bibitem{Tur37a}
Turing, A.M.: On computable numbers, with an application to the {Entscheidungs}
  problem.
\newblock Proceedings of the London Mathematical Society, Series 2 \textbf{42},
  230--265 (1937).
\newblock Correction: \emph{ibid}, 43:544--546, 1937

\bibitem{Wir90a}
Wirsing, M.: Algebraic specification.
\newblock In: J.~van Leeuwen (ed.) Handbook of Theoretical Computer Science,
  vol.~B, pp. 675--788. Elsevier, Amsterdam (1990)

\bibitem{XT96a}
Xia, C., Torrellas, J.: Instruction prefetching of systems codes with layout
  optimized for reduced cache misses.
\newblock In: ISCA '96, pp. 271--282 (1996)

\end{thebibliography}

\end{document}